\documentclass[11pt]{article}

\usepackage{amsmath,amssymb,amsfonts,amsthm}
\usepackage{fullpage}
\usepackage{times}
\newtheorem{definition}{Definition}
\newtheorem{theorem}{Theorem}

\newtheorem{lemma}{Lemma}
\newtheorem{example}{Example}
\newtheorem{remark}{Remark}

\usepackage{tikz}
\usepackage[all]{xy}
\usepackage{multirow}
\usepackage{array}

\usetikzlibrary{fit}
\usetikzlibrary{backgrounds}
\usetikzlibrary{automata}
\usetikzlibrary{shapes}
\usetikzlibrary{matrix}
\usetikzlibrary{fit}

\tikzstyle{State}=[circle, thick, minimum size=0.6cm, inner sep=0cm,draw=black]
\tikzstyle{BState}=[circle, thick, minimum size=0.8cm, inner sep=0cm,draw=black]
\tikzstyle{RState}=[circle, very thick, minimum size=0.8cm, inner sep=0cm,draw=red]
\usetikzlibrary{calc}

\usetikzlibrary{arrows,automata}

\newcommand{\Inf}{\mathrm{Inf}}
\newcommand{\Supp}{\mathrm{Supp}}

\newcommand{\game}{G}
\newcommand{\Act}{A}
\newcommand{\trans}{\delta}
\newcommand{\Obs}{\mathcal{O}}
\newcommand{\obsmap}{\gamma}
\newcommand{\distr}{\mathcal{D}}
\newcommand{\priority}{D}
\newcommand{\straa}{\sigma}
\newcommand{\powset}{\mathcal{P}}
\newcommand{\restr}{\upharpoonright}
\newcommand{\Rec}{\mathsf{Rec}}
\newcommand{\BoolRec}{\mathsf{BoolRec}}
\newcommand{\SetRec}{\mathsf{SetRec}}
\newcommand{\proj}{\mathsf{Proj}}
\newcommand{\ov}{\overline}
\newcommand{\set}[1]{\{#1\}}
\newcommand{\PG}{\mathsf{PrGr}}
\newcommand{\projp}{\mathit{proj}}
\newcommand{\Succ}{\mathsf{Succ}}
\newcommand{\calf}{\mathcal{F}}
\newcommand{\Mem}{\mathsf{Mem}}
\newcommand{\wh}{\widehat}

\title{What is Decidable about Partially Observable Markov Decision Processes with $\omega$-Regular Objectives \\(Full Version)}
\author{Krishnendu Chatterjee \and Martin Chmelik \and Mathieu Tracol}
\date{(IST Austria)} 
\begin{document}

\maketitle



\begin{abstract}
We consider partially observable Markov decision processes (POMDPs) with 
$\omega$-regular conditions specified as parity objectives.
The class of $\omega$-regular languages extends regular languages to 
infinite strings and provides a robust specification language to express
all properties used in verification, and parity objectives are canonical
forms to express $\omega$-regular conditions.
The qualitative analysis problem given a POMDP and a parity 
objective asks whether  there is a strategy to ensure that the objective is 
satisfied with probability~1 (resp. positive probability).
While the qualitative analysis problems are known to be undecidable even for 
very special cases of parity objectives, we establish decidability (with 
optimal complexity) of the qualitative analysis problems for POMDPs with all 
parity objectives under finite-memory strategies.
We establish optimal (exponential) memory bounds 
and EXPTIME-completeness of the qualitative analysis problems under 
finite-memory strategies for POMDPs with parity objectives.
\end{abstract}


\section{Introduction}

\noindent{\bf Partially observable Markov decision processes (POMDPs).}
\emph{Markov decision processes (MDPs)} are standard models for 
probabilistic systems that exhibit both probabilistic 
and nondeterministic behavior~\cite{Howard}.
MDPs have been used to model and solve control problems for stochastic 
systems~\cite{FV97}: 
nondeterminism represents the freedom of the controller to choose a 
control action, while the probabilistic component of the behavior describes the 
system response to control actions. 
In \emph{perfect-observation (or perfect-information) MDPs (PIMDPs)} the 
controller can observe the current state of the system to choose the next 
control actions, whereas in \emph{partially observable MDPs (POMDPs)} the state space is 
partitioned according to observations that the controller can observe 
i.e., given the current state, the controller can only view the observation 
of the state (the partition the state belongs to), but not the precise state~\cite{PT87}.
POMDPs provide the appropriate model to study a wide variety of applications 
such as in computational biology~\cite{Bio-Book}, 
speech processing~\cite{Mohri97}, image 
processing~\cite{IM-Book}, software verification~\cite{CCHRS11}, 
robot planning~\cite{KGFP09}, reinforcement learning~\cite{LearningSurvey}, 
to name a few.
In verification of probabilistic systems, MDPs have been adopted as models 
for concurrent probabilistic 
systems~\cite{CY95},  probabilistic systems operating in open 
environments~\cite{SegalaT}, under-specified probabilistic 
systems~\cite{BdA95}, and applied in diverse domains~\cite{BaierBook,PRISM}.
POMDPs also subsume many other powerful computational models such as 
probabilistic automata~\cite{Rabin63,PazBook} (since probabilistic automata 
(aka blind POMDPS) are a special case of POMDPs where there is only a 
single observation).

\smallskip\noindent{\bf The class of $\omega$-regular objectives.}
An objective specifies the desired set of behaviors (or paths) for the 
controller.
In verification and control of stochastic systems an objective is 
typically an $\omega$-regular set of paths. 
The class of $\omega$-regular languages extends classical regular languages to 
infinite strings, and provides a robust specification language to express
all commonly used specifications, such as safety, reachability, 
liveness, fairness, etc~\cite{Thomas97}.
In a parity objective, every state of the MDP is mapped to a non-negative 
integer priority (or color) and the goal is to ensure that the minimum 
priority visited infinitely often is even.
Parity objectives are a canonical way to define such $\omega$-regular 
specifications (e.g., all specifications in verification expressed as a 
linear-time temporal logic (LTL) formula can be translated to a parity 
objective).
Thus POMDPs with parity objectives provide the theoretical framework to 
study problems such as the verification and control of stochastic systems.

\smallskip\noindent{\bf Qualitative and quantitative analysis.} 
The analysis of POMDPs with parity objectives can be classified into  
qualitative and quantitative analysis. 
Given a POMDP with a parity objective and a start state, the 
\emph{qualitative analysis} asks whether the objective can be ensured with 
probability~1 (\emph{almost-sure winning}) or 
positive probability (\emph{positive winning});
whereas the \emph{quantitative analysis} asks whether the objective can 
be satisfied with probability at least $\lambda$ for a given threshold 
$\lambda \in (0,1)$.

\smallskip\noindent{\bf Importance of qualitative analysis.} 
The qualitative analysis of MDPs is an important problem in verification 
that is of interest independent of the quantitative analysis problem.
There are many applications where we need to know whether the correct 
behavior arises with probability~1.
For instance, when analyzing a randomized embedded scheduler, we are
interested in whether every thread progresses with probability~1~\cite{EMSOFT05}.
Even in settings where it suffices to satisfy certain specifications with 
probability $\lambda<1$, the correct choice of $\lambda$ is a challenging 
problem, due to the simplifications introduced during modeling.
For example, in the analysis of randomized distributed algorithms it is 
quite common to require correctness with probability~1 
(see, e.g.,~\cite{PSL00,KNP_PRISM00,Sto02b}). 
Furthermore, in contrast to quantitative analysis, 
qualitative analysis is robust to numerical perturbations and modeling errors in the 
transition probabilities.
Thus qualitative analysis of POMDPs with parity objectives is a 
fundamental theoretical problem in verification and analysis of 
probabilistic systems.

\smallskip\noindent{\bf Previous results.}
On one hand POMDPs with parity objectives provide a rich framework to model 
a wide variety of practical problems, on the other hand, most theoretical 
results established for POMDPs are \emph{negative} (undecidability) results.
There are several deep undecidability results established for the 
special case of probabilistic automata (that immediately imply undecidability
for the more general case of POMDPs).
The basic undecidability results are for probabilistic automata over
finite words (that can be considered as a special case of parity objectives).
The quantitative analysis problem is undecidable for probabilistic automata 
over finite words~\cite{Rabin63,PazBook,CL89}; and it was shown in~\cite{MHC03} 
that even the following approximation version is undecidable: for 
any fixed $0<\epsilon <\frac{1}{2}$, given a probabilistic automaton and 
the guarantee that either (a)~there is a word accepted with probability at 
least $1-\epsilon$; or 
(ii)~all words are accepted with probability at most $\epsilon$;
decide whether it is case~(i) or case~(ii).
The almost-sure (resp. positive) problem for probabilistic automata over 
finite words reduces to the non-emptiness question of universal 
(resp. non-deterministic) automata over finite words and 
is PSPACE-complete (resp. solvable in polynomial time).
However, another related decision question whether for every $\epsilon>0$
there is a word that is accepted with probability at least 
$1-\epsilon$ (called the value~1 problem) is undecidable for probabilistic 
automata over finite words~\cite{GO10}.
Also observe that all undecidability results for probabilistic automata over
finite words carry over to POMDPs where the controller is restricted to 
finite-memory strategies.
In~\cite{Meu99}, the authors consider POMDPs with finite-memory strategies 
under expected rewards, but the general problem remains undecidable.
For qualitative analysis of POMDPs with parity objectives, deep undecidability 
results were established even for very special cases of parity objectives 
(even in the special case of probabilistic automata).
It was shown in~\cite{BBG08,BGB12} that the almost-sure (resp. positive) problem is 
undecidable for probabilistic automata with coB\"uchi (resp. B\"uchi) objectives
which are special cases of parity objectives that use only two priorities.
In summary the most important theoretical results are negative in the 
sense that they establish undecidability results.

\smallskip\noindent{\bf Our contributions.}
For POMDPs with parity objectives, 
all questions related to quantitative analysis are undecidable, and the 
qualitative analysis problems are also undecidable in general. 
However, the undecidability proofs for the qualitative analysis of 
POMDPs with parity objectives crucially require the use of 
\emph{infinite-memory} strategies for the controller. 
In all practical applications, 
the controller must be a \emph{finite-state} controller to be implementable.
Thus for all practical purposes the relevant question is the existence of 
finite-memory controllers.
The quantitative analysis problem remains undecidable even under finite-memory
controllers as the undecidability results are established for
probabilistic automata over finite words.
In this work we study the most prominent remaining theoretical open question
(that is also of practical relevance)  for POMDPs with parity objectives that 
whether the qualitative analysis of POMDPs 
with parity objectives is decidable or undecidable for finite-memory 
strategies (i.e., finite-memory controllers).
Our main result is the \emph{positive} result that qualitative 
analysis of POMDPs with parity objectives is \emph{decidable} under
finite-memory strategies.
Moreover, for qualitative analysis of POMDPs with parity objectives under
finite-memory strategies  
we establish optimal complexity bounds both for strategy 
complexity as well as computational complexity.
The details of our contributions are as follows:
\begin{enumerate}
\item \emph{(Strategy complexity).}
Our first result shows that \emph{belief-based} strategies 
are not sufficient (where a belief-based strategy is based
on the subset construction that remembers the possible set 
of current states): we show that there exist POMDPs with coB\"uchi 
objectives where finite-memory almost-sure winning strategy 
exists but there exists no randomized belief-based almost-sure winning 
strategy. 
All previous results about decidability for almost-sure winning in sub-classes 
of POMDPs crucially relied on the sufficiency of randomized belief-based 
strategies that allowed standard 
techniques like subset construction to establish decidability.
However, our counter-example shows that previous techniques based on 
simple subset construction (to construct an exponential size PIMDP) are not adequate to solve 
the problem.   
Before the result for parity objectives, we consider a slightly more general
form of objectives, called Muller objectives.
For a Muller objective a set $\calf$ of subsets of colors is given and 
the set of colors visited infinitely often must belong to $\calf$.
We show our main result that given a POMDP with $|S|$ 
states and a Muller objective with $d$ colors (priorities), if there 
is a finite-memory almost-sure (resp. positive) winning strategy, 
then there is an almost-sure (resp. positive) winning 
strategy that uses at most $\Mem^*=2^{2\cdot |S|} \cdot (2^{2^d})^{|S|}$ memory.
Developing on our result for Muller objectives, 
for POMDPs with parity objectives we show that if there is a 
finite-memory almost-sure (resp. positive) winning strategy, 
then there is an almost-sure (resp. positive) winning 
strategy that uses at most $2^{3\cdot d\cdot |S|}$ memory.
Our exponential memory upper bound for parity objectives is 
optimal as it has been already established in~\cite{CDH10a} that 
almost-sure winning strategies require at least exponential
memory even for the very special case of reachability objectives 
in POMDPs.

\item \emph{(Computational complexity).} 
We present an exponential time algorithm for the qualitative analysis of
POMDPs with parity objectives under finite-memory strategies, and thus 
obtain an EXPTIME upper bound.
The EXPTIME-hardness follows from~\cite{CDH10a} for the special case of
reachability and safety objectives, and thus we obtain the 
optimal EXPTIME-complete computational complexity result.
\footnote{ Recently, Nain and Vardi (personal communication, to appear 
LICS 2013) considered the finite-memory strategies problem for one-sided
partial-observation games and established 2EXPTIME upper bound.
Our work is independent and establishes optimal (EXPTIME-complete) 
complexity bounds for POMDPs.}
\end{enumerate}
In Table~\ref{tab:almost_str} and Table~\ref{tab:almost_comp} we summarize the 
results for strategy complexity and computational complexity, respectively.

\smallskip\noindent{\em Technical contributions.}
The key technical contribution for the decidability result is as follows.
Since belief-based strategies are not sufficient, standard subset 
construction techniques do not work.
For an arbitrary finite-memory strategy we construct a projected 
strategy that collapses memory states based on a projection graph 
construction given the strategy. 
The projected strategy at a collapsed memory state plays uniformly 
over actions that were played at all the corresponding memory states 
of the original strategy.
The projected strategy thus plays more actions with positive
probability.
The key challenge is to show the bound on the size of
the projection graph, and to show that the projected strategy,
even though plays more actions, does not destroy the structure 
of the recurrent classes of the original strategy.
For parity objectives, we show a reduction from general parity objectives
to parity objectives with two priorities on a polynomially larger
POMDP and from our general result for Muller objectives obtain the 
optimal memory complexity bounds for parity objectives.
For the computational complexity result, we show how to construct an
exponential size special class of POMDPs (which we call belief-observation
POMDPs where the belief is always the current observation) and
present polynomial time algorithms for the qualitative analysis of
the special belief-observation POMDPs of our construction.


\begin{table}[h]
\begin{center}
\resizebox{\linewidth}{!}{
\begin{tabular}{|c|>{\centering\arraybackslash}p{2.0cm}|>{\centering\arraybackslash}p{2.9cm}|>{\centering\arraybackslash}p{2.0cm}|>{\centering\arraybackslash}p{3.3cm}|>{\centering\arraybackslash}p{1.6cm}|>{\centering\arraybackslash}p{1.7cm}|}
\hline
\multirow{2}{*}{Objectives} & \multicolumn{2}{c|}{Almost-sure } &
\multicolumn{2}{c|}{Positive}
&\multicolumn{2}{c|}{Quantitative}  \\
\cline{2-7}
 & Inf. Mem. &  Fin. Mem. & Inf. Mem. &  Fin. Mem. & Inf. Mem. &  Fin. Mem.  \\
\hline
\hline
\multirow{2}{*}{B\"uchi}  & &  & & UB: {\bf Exp. ${\mathbf{2^{6\cdot |S|}}}$} & & \\
& Exp. (belief-based) & Exp. (belief-based) & Inf. mem.  req. & LB: Exp. \hspace{1cm} {\bf
(belief not sufficient)}  & Inf. mem. req. & No bnd.\\

\hline
\multirow{2}{*}{coB\"uchi} & & UB: {\bf Exp.  ${\mathbf{2^{6\cdot |S|}}}$ } & UB: Exp. & UB: Exp. & &\\
&  Inf. mem. req.& LB: Exp. \hspace{1cm} {\bf (belief not sufficient)} & LB: Exp. {\bf (belief not
sufficient)} & LB: Exp.\hspace{1cm} {\bf (belief not sufficient)} & Inf. mem. req. & No bnd.\\
\hline
\multirow{2}{*}{Parity}    & & UB: {\bf Exp.
${\mathbf{2^{3 \cdot d \cdot |S|}}}$} &  & UB: {\bf Exp.
${\mathbf{2^{3 \cdot d \cdot |S|}}}$} & & \\
& Inf. mem. req. & LB: Exp. \hspace{1cm} {\bf (belief not sufficient)} & Inf.
mem. req.& LB: Exp. \hspace{1cm} {\bf
(belief not sufficient)}  & Inf. mem. req.& No  bnd.\\
\hline
\hline
\end{tabular}
}
\end{center}

\caption{Strategy complexity for POMDPs with parity objectives, where $|S|$ is the size of state space, and $d$ the number of 
priorities, (UB denotes upper bound and LB denotes lower bound). 
The results in boldface are new results included in the present paper.}\label{tab:almost_str}
\end{table}

\begin{table}[h]
\begin{center}
\resizebox{\linewidth}{!}{
\begin{tabular}{|c|>{\centering\arraybackslash}p{2.4cm}|>{\centering\arraybackslash}p{2.4cm}|>{\centering\arraybackslash}p{2.4cm}|>{\centering\arraybackslash}p{2.4cm}|>{\centering\arraybackslash}p{1.6cm}|>{\centering\arraybackslash}p{1.6cm}|}
\hline
\multirow{2}{*}{Objectives} & \multicolumn{2}{c|}{Almost-sure } &
\multicolumn{2}{c|}{Positive}
&\multicolumn{2}{c|}{Quantitative}  \\
\cline{2-7}
 & Inf. Mem. &  Finite Mem. & Inf. Mem. &  Finite Mem. & Inf. Mem. &  Finite Mem.  \\
\hline
\hline
B\"uchi  & EXP-complete & EXP-complete & Undec. & {\bf EXP-complete} & Undec. & Undec.\\
\hline
coB\"uchi & Undec. & {\bf EXP-complete} & EXP-complete & EXP-complete &
Undec. & Undec.\\
\hline
Parity    & Undec. & {\bf EXP-complete} & Undec. & {\bf EXP-complete}
& Undec. & Undec.\\

\hline
\end{tabular}
}
\end{center}
\caption{Computational complexity for POMDPs with parity objectives.
The results in boldface are new results included in the present paper.}\label{tab:almost_comp}

\end{table}

\newcommand{\Prb}{\mathbb{P}}
\newcommand{\Cone}{\mathsf{Cone}}
\newcommand{\Last}{\mathsf{Last}}
\newcommand{\Safe}{\mathsf{Safe}}
\newcommand{\Reach}{\mathsf{Reach}}
\newcommand{\Buchi}{\mathsf{Buchi}}
\newcommand{\coBuchi}{\mathsf{coBuchi}}
\newcommand{\Parity}{\mathsf{Parity}}
\newcommand{\Muller}{\mathsf{Muller}}
\newcommand{\col}{\mathsf{col}}

\newcommand{\Outcome}{{\mathsf{Outcome}}}
\newcommand{\target}{{\cal T}}
\newcommand{\nats}{\mathbb{N}} \newcommand{\Nats}{\mathbb{N}}
\newcommand{\reals}{\mathbb{R}} \newcommand{\Reals}{\mathbb{R}}
\newcommand{\nat}{\mathbb N} 
\newcommand{\rat}{{\mathbb Q}}
\newcommand{\D}{\mathbf{\mathfrak{D}}}
\newcommand{\C}{\mathcal{C}}
\newcommand{\B}{\mathcal{B}}

\newcommand{\slopefrac}[2]{\leavevmode\kern.1em
  \raise .5ex\hbox{\the\scriptfont0 #1}\kern-.1em
  /\kern-.15em\lower .25ex\hbox{\the\scriptfont0 #2}}
\newcommand{\half}{\slopefrac{1}{2}}

\section{Definitions}
In this section we present the basic definitions of POMDPs, strategies 
(policies), $\omega$-regular objectives, and the winning modes.

\smallskip\noindent{\bf Notations.}
Given a finite set $X$, we denote by $\powset(X)$ the set of subsets of $X$,
i.e., $\powset(X)$ is the power set of $X$.
A probability distribution $f$ on $X$ is a function $f:X \to [0,1]$ such 
that $\sum_{x\in X} f(x)=1$, and we denote by  $\distr(X)$ the set of 
all probability distributions on $X$. For $f \in \distr(X)$ we denote by $\Supp(f)=\set{x\in X \mid f(x)>0}$
the support of $f$.

\begin{definition}[POMDP]
A \emph{Partially Observable Markov Decision Process (POMDP)} is a 
tuple $\game=(S,\Act,\trans,\Obs,\obsmap,s_0)$ where:
\begin{itemize}
 \item $S$ is a finite set of states;
 \item $\Act$ is a finite alphabet of \emph{actions};
 \item $\trans:S\times\Act \rightarrow \distr(S)$ is a 
 \emph{probabilistic transition function} that given a state $s$ and an
 action $a \in \Act$ gives the probability distribution over the successor 
 states, i.e., $\trans(s,a)(s')$ denotes the transition probability from state
 $s$ to state $s'$ given action $a$; 
 \item $\Obs$ is a finite set of \emph{observations}; 
 \item $\obsmap:S\rightarrow \Obs$ is an \emph{observation function} that 
  maps every state to an observation; and 
 \item $s_0$ is the initial state. 
\end{itemize}
\end{definition}
\noindent
Given $s,s'\in S$ and $a\in\Act$, we also write $\trans(s'|s,a)$ for 
$\trans(s,a)(s')$.
For an observation $o$, 
we denote by $\obsmap^{-1}(o)=\set{s \in S \mid \obsmap(s)=o}$
the set of states with observation $o$.
For a set $U \subseteq S$ of states and $O \subseteq \Obs$ of observations 
we denote 
$\obsmap(U)=\set{o \in \Obs \mid \exists s \in U. \ \obsmap(s)=o}$ 
and $\obsmap^{-1}(O)= \bigcup_{o \in O} \obsmap^{-1}(o)$.

\begin{remark}
For technical convenience we have assumed that there is an unique initial 
state and we will also assume that the initial state $s_0$ has a unique 
observation, i.e., $|\obsmap^{-1}(\obsmap(s_0))|=1$.
In general there is an initial distribution $\alpha$ over initial states
that all have the same observation, i.e., $\Supp(\alpha) \subseteq 
\obsmap^{-1}(o)$, for some $o \in \Obs$.
However, this can be modeled easily by adding a new initial state 
$s_{\mathit{new}}$ with a unique observation such that in the first step gives 
the desired initial probability distribution $\alpha$, 
i.e., $\trans(s_{\mathit{new}},a)=\alpha$ for all actions $a \in \Act$.
Hence for simplicity we assume there is a unique initial state $s_0$ with a
unique observation.
\end{remark}

\smallskip\noindent{\bf Plays, cones and belief-updates.}
A \emph{play} (or a path) in a POMDP is an 
infinite sequence $(s_0,a_0,s_1,a_1,s_2,a_2,\ldots)$ of states and 
actions such that for all $i \geq 0$ we have $\trans(s_i,a_i)(s_{i+1})>0$.  
We write $\Omega$ for the set of all plays.
For a finite prefix $w \in (S\cdot A)^* \cdot S$ of a play, we denote by $\Cone(w)$ the 
set of plays with $w$ as the prefix (i.e., the cone or cylinder of the prefix $w$), 
and denote by $\Last(w)$ the last state of $w$.
For a finite prefix $w=(s_0,a_0,s_1,a_1,\ldots,s_n)$ 
we denote by 
$\obsmap(w)=(\obsmap(s_0),a_0,\obsmap(s_1),a_1,\ldots,\obsmap(s_n))$ 
the observation and action sequence associated with $w$.
For a finite sequence $\rho=(o_0,a_0,o_1,a_1,\ldots,o_n)$ of observations and actions, the \emph{belief} $\B(\rho)$ 
after the prefix $\rho$ is the set of states in which a finite prefix 
of a play can be after the sequence $\rho$ of observations and actions, 
i.e., $\B(\rho)=\set{s_n=\Last(w) \mid w=(s_0,a_0,s_1,a_1,\ldots,s_n), w \mbox{
is a prefix of a play, and for all } 0\leq i \leq n. \; \obsmap(s_i)=o_i}$. 
The belief-updates associated with finite-prefixes are as follows:
for prefixes $w$ and $w'=w \cdot a \cdot s$ the belief update is 
defined inductively as 
$\B(\obsmap(w')) = 
\left(\bigcup_{s_1 \in \B(\obsmap(w))} \Supp(\trans(s_1,a)) \right)
\cap \obsmap^{-1}(s)$, 
i.e., the set
$\left(\bigcup_{s_1 \in \B(\obsmap(w))} \Supp(\trans(s_1,a)) \right)$
denotes the possible successors given the belief $\B(\obsmap(w))$ and 
action $a$, and then the intersection with the set of states with the 
current observation $\gamma(s)$ gives the new belief set.

\smallskip\noindent{\bf Strategies.}
A \emph{strategy (or a policy)} is a recipe to extend prefixes of plays and 
is a function $\sigma: (S\cdot A)^* \cdot S \to \distr(A)$ that given a finite 
history (i.e., a finite prefix of a play) selects a probability distribution 
over the actions.
Since we consider POMDPs, strategies are \emph{observation-based}, i.e., 
for all histories $w=(s_0,a_0,s_1,a_1,\ldots,a_{n-1},s_n)$ and 
$w'=(s_0',a_0,s_1',a_1,\ldots,a_{n-1},s_n')$ such that for all 
$0\leq i \leq n$ we have $\obsmap(s_i)=\obsmap(s_i')$ (i.e., $\obsmap(w) = \obsmap(w')$), we must have 
$\sigma(w)=\sigma(w')$.
In other words, if the observation sequence is the same, then the strategy 
cannot distinguish between the prefixes and must play the same. 
We now present an equivalent definition of observation-based strategies  
such that the memory of the strategy is explicitly specified, and 
will be required to present finite-memory strategies.

\begin{definition}[Strategies with memory and finite-memory strategies]
A \emph{strategy} with memory is a tuple $\sigma=(\sigma_u,\sigma_n,M,m_0)$ 
where:
\begin{itemize}
 \item \emph{(Memory set).} $M$ is a denumerable set (finite or infinite) of memory elements (or memory states).
 \item \emph{(Action selection function).} The function $\sigma_n:M\rightarrow \distr(\Act)$ is the 
	\emph{action selection function} that given the current memory 
	state gives the probability distribution over actions.
 \item \emph{(Memory update function).} The function $\sigma_u:M\times\Obs\times\Act\rightarrow \distr(M)$ 
 is the \emph{memory update function} that given the current memory state, 
 the current observation and action, updates the memory state probabilistically.
 \item \emph{(Initial memory).} The memory state $m_0\in M$ is the initial memory state.
\end{itemize}
A strategy is a \emph{finite-memory} strategy if the set $M$ of memory elements is finite.
A strategy is \emph{pure (or deterministic)} if the memory update function 
and the action selection function are deterministic, i.e., 
$\sigma_u: M \times \Obs \times \Act \to M$ and $\sigma_n: M \to \Act$.
A strategy is \emph{memoryless (or stationary)} if it is independent of the 
history but depends only on the current observation, and can be represented
as a function $\sigma: \Obs \to \distr(\Act)$.  
\end{definition}
\begin{remark}
It was shown in~\cite{CDGH10} that in POMDPs pure strategies are as powerful 
as randomized strategies, hence in sequel we omit discussions about 
pure strategies.
\end{remark}
\smallskip\noindent{\bf Probability measure.}
Given a strategy $\sigma$, the unique probability measure obtained given $\sigma$ is denoted as $\mathbb{P}^{\sigma}(\cdot)$.
We first define the measure $\mu^\sigma(\cdot)$ on cones.
For $w=s_0$ we have $\mu^\sigma(\Cone(w))=1$, and 
for $w=s$ where $s\neq s_0$ we have  $\mu^\sigma(\Cone(w))=0$; and 
for $w' = w \cdot a\cdot s$ 
we have 
$\mu^\sigma(\Cone(w'))= \mu^\sigma(\Cone(w)) \cdot \sigma(w)(a) \cdot \trans(\Last(w),a)(s)$. 
By Carathe\'odary's extension theorem, the function $\mu^\sigma(\cdot)$
can be uniquely extended to a probability measure $\Prb^{\sigma}(\cdot)$
over Borel sets of infinite plays~\cite{Billingsley}.

\smallskip\noindent{\bf Objectives.}
An \emph{objective} in a POMDP $G$ is a Borel set $\varphi \subseteq \Omega$ 
of plays in the Cantor topology on $\Omega$~\cite{Kechris}. 
All objectives we consider in this paper lie in the first $2\half$-levels of 
the Borel hierarchy. 
We specifically consider the parity objective, which is a canonical form to express all 
$\omega$-regular objectives~\cite{Thomas97}.
Thus parity objectives provide a robust specification language to express all commonly 
used properties in verification and system analysis.  
For a play $\rho = (s_0, a_0, s_1, a_1, s_2 \ldots)$, we denote 
by $\Inf(\rho) = \set{ s \in S \mid \forall i \geq 0 \cdot \exists j \geq i: s_j = s }$ 
the set of states that occur infinitely often in $\rho$.
We consider the following objectives. 

\begin{itemize}
 	\item \emph{Reachability and safety objectives.}
	Given a set $\target \subseteq S$ of target states, the \emph{reachability} objective 
	$\Reach(\target) = \{ (s_0, a_0, s_1, a_1, s_2 \ldots) \in \Omega \mid \exists k \geq 0:  s_k \in \target\}$
	requires that a target state in $\target$ is visited at least once.
	Dually, the \emph{safety} objective $\Safe(\target) = \{ (s_0, a_0, s_1, a_1, s_2 \ldots) \in \Omega \mid \forall k \geq 0:  s_k \in \target\}$ 
	requires that only states in $\target$ are visited.

	\item \emph{B\"uchi and coB\"uchi objectives.}
	Given a set $\target \subseteq S$ of target states,
	the \emph{B\"uchi} objective $\Buchi(\target) = \{ \rho \in \Omega \mid \Inf(\rho) \cap \target \neq \emptyset\}$ 
	requires that a state in $\target$ is visited infinitely often. 
	Dually, the \emph{coB\"uchi} objective $\coBuchi(\target) =\{\rho \in \Omega \mid \Inf(\rho) \subseteq \target\}$ 
	requires that only states in $\target$ are visited infinitely often.
	
	\item\emph{Parity objectives.} 
	For $d \in \Nats$, let $p:S \to \set{0,1,\ldots,d}$ be a 
	\emph{priority function} that maps each state 
	to a non-negative integer priority.
	The \emph{parity} objective $\Parity(p) = \{\rho \in \Omega \mid \min\set{ p(s) \mid s \in \Inf(\rho)} 
	\text{ is even} \}$ requires that the smallest priority that appears
infinitely often is even.

	\item\emph{Muller objectives.} Let $\priority$ be a set of colors, and 
	$\col:S \to \priority$ be a color mapping function that maps every state
to a color.
	A Muller objective $\calf$ consists of a set of subsets of colors and requires 
	that the set of colors visited infinitely often belongs to $\calf$, i.e.,
	$\calf \in \powset(\powset(\priority))$ and 
	$\Muller(\calf)=\set{\rho \in \Omega \mid \set{\col(s) \mid s \in 
      \Inf(\rho)} \in \calf}$

\end{itemize}
Note that a reachability objective $\Reach(\target)$ can be viewed as a special 
case of B\"uchi as well as coB\"uchi objectives, (assuming w.l.o.g. that all target states $s \in \target$ are \emph{absorbing}, 
i.e., $\trans(s,a)(s) = 1$ for all $a \in \Act$)
and analogously safety objectives are also special cases of B\"uchi and coB\"uchi objectives.
The objectives $\Buchi(\target)$ and $\coBuchi(\target)$ 
are special cases of parity objectives defined by respective priority functions
$p_1,p_2$ such that $p_1(s) = 0$ and $p_2(s) = 2$ if $s \in \target$,
and $p_1(s) = p_2(s) = 1$ otherwise. 
Given a set $U\subseteq S$ we will denote by $p(U)$ the set of priorities of the set $U$ 
given by the priority function $p$, i.e., $p(U)= \set{p(s) \mid s \in U}$,
and similarly $\col(U)=\set{\col(s) \mid s \in U}$. 
Also observe that parity objectives are a special case of Muller objectives,
however, given a POMDP with a Muller objective with color set $\priority$, 
an equivalent POMDP with $|S|\cdot |\priority|!$ states and a parity objective 
with $|\priority|^2$ priorities can be constructed using the well-known latest appearance 
record (LAR) construction of~\cite{GH82} for conversion of Muller objectives
to parity objectives.
An objective $\varphi$ is visible if for all plays $\rho$ and $\rho'$ that
have the same observation sequence we have $\rho\in \varphi$ iff $\rho' \in 
\varphi$.

\smallskip\noindent{\bf Winning modes.}
Given a POMDP, an objective $\varphi$, and a class $\C$ of strategies, we say that:

\begin{itemize}
\item a strategy $\sigma \in \C$ is \emph{almost-sure winning} 
if $\Prb^{\sigma}(\varphi) = 1$;

\item a strategy $\sigma \in \C$ is \emph{positive winning} 
if $\Prb^{\sigma}(\varphi) >0$;

\item the POMDP is \emph{limit-sure winning} if for all $\varepsilon>0$
there exists a strategy   $\sigma \in \C$ for player~$1$ such that  
$\Prb^{\sigma}(\varphi) \geq 1 - \epsilon$; and

\item a strategy $\sigma \in \C$ is \emph{quantitative winning}, for a 
threshold $\lambda \in (0,1)$, if $\Prb^{\sigma}(\varphi) \geq \lambda$.

\end{itemize}
We first precisely summarize related works in the following Theorem.

\begin{theorem}[Decidability and complexity under general 
strategies~\cite{Rabin63,PazBook,CL89,GO10,CH10,BBG08,BGB12,Reif79,Reif84,CDH10a}]
The following assertions hold for POMDPs with the class $\C$ of all
infinite-memory (randomized or pure) strategies:
\begin{enumerate}
\item The quantitative winning problem is undecidable for safety, 
reachability, B\"uchi, coB\"uchi, parity, and Muller objectives.

\item The limit-sure winning problem is EXPTIME-complete for safety 
objectives; and undecidable for reachability, B\"uchi, 
coB\"uchi, parity, and Muller objectives.

\item The almost-sure winning problem is EXPTIME-complete for safety, 
reachability, and B\"uchi objectives; and 
undecidable for coB\"uchi, parity, and Muller objectives.

\item The positive winning problem is PTIME-complete for reachability 
objectives, EXPTIME-complete for safety and coB\"uchi objectives; and 
undecidable for B\"uchi, parity, and Muller objectives.

\end{enumerate}
\end{theorem}

\noindent{\em Explanation of the previous results and implications 
under finite-memory policies.}
All the undecidability results follow from the special case of 
probabilistic automata:
the undecidability of the quantitative problem for probabilistic 
automata follows from~\cite{Rabin63,PazBook,CL89}; 
the undecidability of the limit-sure winning for finite words and 
reachability objectives was established in~\cite{GO10,CH10} 
(the undecidability of limit-sure reachability also implies 
undecidability for B\"uchi, coB\"uchi and parity objectives);
the undecidability for positive winning for B\"uchi and 
almost-sure winning for coB\"uchi objectives was established
in~\cite{BBG08,BGB12}.
For the decidable results, the optimal complexity results 
for safety objectives can be obtained from the results 
of~\cite{Reif79,Reif84} and all the other results follow 
from~\cite{CDH10a,BGB12}.
If the classes of strategies are restricted to finite-memory 
strategies, then the undecidability results for quantitative 
winning and limit-sure winning still hold, as they are established
for reachability objectives and for reachability objectives
finite-memory suffices. 
The most prominent and important open question is whether the almost-sure 
and positive winning problems are decidable for parity and Muller
objectives in POMDPs under finite-memory strategies.
All the lower bounds (i.e., hardness and undecidability) results have been
established for the cases when the objectives are restricted
to be visible.

\section{Strategy Complexity for Muller Objectives under Finite-memory Strategies}
\label{sec:finmem}
In this section we will first show that belief-based stationary strategies are 
not sufficient for finite-memory almost-sure winning strategies in POMDPs with 
coB\"uchi objectives;
and then present the upper bound on memory size required for finite-memory 
almost-sure and positive winning strategies in POMDPs with Muller objectives. 
Our proofs will use many basic results on Markov chains 
and we start with them in the following subsection.

\subsection{Basic properties of Markov chains}
Since our proof relies heavily on Markov chains we start with some basic 
definitions and properties related to Markov chains that are essential for our proofs.

\smallskip\noindent{\bf Markov chains, recurrent classes, and reachability.} 
A Markov chain $\ov{\game}=(\ov{S},\ov{\trans})$ consists of a \emph{finite} set $\ov{S}$ of 
states and a probabilistic transition function $\ov{\trans}:\ov{S} \rightarrow \distr(\ov{S})$.
Given the Markov chain, we consider the directed graph $(\ov{S},\ov{E})$ where $\ov{E}=\set{(\ov{s},\ov{s}') 
\mid \trans(\ov{s}' \mid \ov{s}) >0}$.
A \emph{recurrent class} $\ov{C} \subseteq \ov{S}$ of the Markov chain is a bottom 
strongly connected component (scc) in the graph  $(\ov{S},\ov{E})$ 
(a bottom scc is an scc with no edges out of the scc).
We denote by $\Rec(\ov{\game})$ the set of recurrent classes of the Markov chain,
i.e., $\Rec(\ov{\game})=\set{\ov{C} \mid \ov{C} \text{ is a recurrent class}}$.
Given a state $\ov{s}$ and a set $\ov{U}$ of states, we say that $\ov{U}$ is
reachable from $\ov{s}$ if there is a path from $\ov{s}$ to some state in 
$\ov{U}$ in the graph $(\ov{S},\ov{E})$. 
Given a state $\ov{s}$ of the Markov chain we denote by $\Rec(\ov{\game})(\ov{s}) \subseteq \Rec(\ov{\game})$ 
the subset of the recurrent classes reachable from $\ov{s}$ in $\ov{G}$.
A state is recurrent if it belongs to a recurrent class.
The following standard properties of reachability and the recurrent classes 
will be used in our proof:
\begin{enumerate}
\item \emph{Property~1.}
(a)~For a set $\ov{T} \subseteq \ov{S}$, if for all states $\ov{s} \in \ov{S}$ 
there is a path to $\ov{T}$ (i.e., for all states there is a positive
probability to reach $\ov{T}$), then from all states the set $\ov{T}$
is reached with probability~1.
(b)~For all states $\ov{s}$, if the Markov chain starts at $\ov{s}$, then 
the set ${\mathcal C}=\bigcup_{\ov{C} \in \Rec(\ov{\game})(\ov{s})} \ov{C}$ is
reached with probability~1, 
i.e., the set of recurrent classes  are reached with probability~1.

\item \emph{Property~2.}
If $\ov{s}$ is recurrent and it belongs to a recurrent class $\ov{C}$,
then $\Rec(\ov{\game})(\ov{s})=\set{\ov{C}}$.

\item \emph{Property~3.}
For a recurrent class $\ov{C}$, for all states $\ov{s} \in \ov{C}$,
if the Markov chain starts at $\ov{s}$, then all states
$\ov{t}\in \ov{C}$ are visited infinitely often with probability~1.

\item \emph{Property~4.}
If $\ov{s}'$ is reachable from $\ov{s}$, then 
$\Rec(\ov{\game})(\ov{s}') \subseteq \Rec(\ov{\game})(\ov{s})$.

\item \emph{Property~5.} For all $\ov{s}$ we have 
$\Rec(\ov{\game})(\ov{s}) =\bigcup_{(\ov{s},\ov{s}') \in \ov{E}} 
\Rec(\ov{\game})(\ov{s}')$. 
\end{enumerate}
The following lemma is an easy consequence of the above properties.

\begin{lemma}\label{lemm:Markov-basic}
Given a Markov chain $\ov{G}=(\ov{S},\ov{\trans})$ with Muller objective 
$\Muller(\calf)$ (or a parity objective $\Parity(p)$), 
a state $\ov{s}$ is almost-sure winning (resp. positive winning)
if for all  recurrent classes $\ov{C} \in \Rec(\ov{G})(\ov{s})$ 
(resp. for some recurrent class $\ov{C} \in \Rec(\ov{G})(\ov{s})$) reachable from 
$\ov{s}$ we have $\col(\ov{C}) \in \calf$ ($\min(p(\ov{C}))$ is even for the parity objective).  
\end{lemma}
\begin{proof}
From $\ov{s}$ the set of recurrent classes reachable from $\ov{s}$ is  
reached with probability~1 (Property~1~(b)), and every recurrent class reachable is 
reached with positive probability. 
In every recurrent class $\ov{C}$ the minimum priority visited 
infinitely often with probability~1 is the minimum priority of 
$\ov{C}$ (Property~3).
Also in every recurrent class $\ov{C}$ the set of colors visited
infinitely often with probability~1 is exactly the set $\col(\ov{C})$ (Property~3).
The desired result follows. 
\end{proof}

\smallskip\noindent{\bf Markov chains $\game \restr \straa$ under finite memory strategies $\straa$.}
We now define Markov chains obtained by fixing  a finite-memory strategy in a POMDP
$\game$. 
A finite-memory strategy $\straa=(\straa_u,\straa_n,M,m_0)$ induces a finite-state Markov chain 
$(S\times M,\trans_{\straa})$, denoted $\game \restr \straa$, 
with the probabilistic transition function $\trans_\straa: S\times M \rightarrow \distr(S\times M)$: 
given $s,s'\in S$ and $m,m'\in M$, the transition $\trans_\straa\big((s',m')\ |\ (s,m)\big)$ is the probability to go from state 
$(s,m)$ to state $(s',m')$ in one step under the strategy $\straa$. 
The probability of transition can be decomposed as follows:
\begin{itemize}
 \item First an action $a\in\Act$ is sampled according to the distribution $\straa_n(m)$; 
 \item then the next state $s'$ is sampled according to the distribution $\trans(s,a)$; and
 \item finally the new memory $m'$ is sampled according to the distribution $\straa_u(m,\obsmap(s'),a)$
 (i.e., the new memory is sampled according $\straa_u$ given the old memory, new observation and the action).
\end{itemize}
More formally, we have:
\[
\trans_\sigma\big((s',m')\ |\
(s,m)\big)=\sum_{a\in\Act}\straa_n(m)(a)\cdot\trans(s,a)(s')\cdot\straa_u(m,\obsmap(s'),a)(m').
\]
Given $s\in S$ and $m\in M$, we write $(G \restr \straa)_{(s,m)}$ for the finite state 
Markov chain induced on $S\times M$ by the transition function 
$\trans_{\straa}$, given the initial state is $(s,m)$. 


\subsection{Belief-based stationary strategies are not sufficient}
For all previous decidability results for almost-sure winning in POMDPs, the 
key was to show that \emph{belief-based stationary} strategies are sufficient. 
A strategy is belief-based stationary if its memory relies only on the subset 
construction where the subset denotes the possible current states, 
i.e., the strategy plays only depending on the set of possible
current states of the POMDP, which is called \emph{belief}.
In POMDPs with B\"uchi objectives, belief-based stationary
 strategies are sufficient for almost-sure winning.
We now show with an example that there exist POMDPs with coB\"uchi objectives,
where finite-memory randomized almost-sure winning strategies exist, but there 
exists no belief-based stationary almost-sure winning strategy.

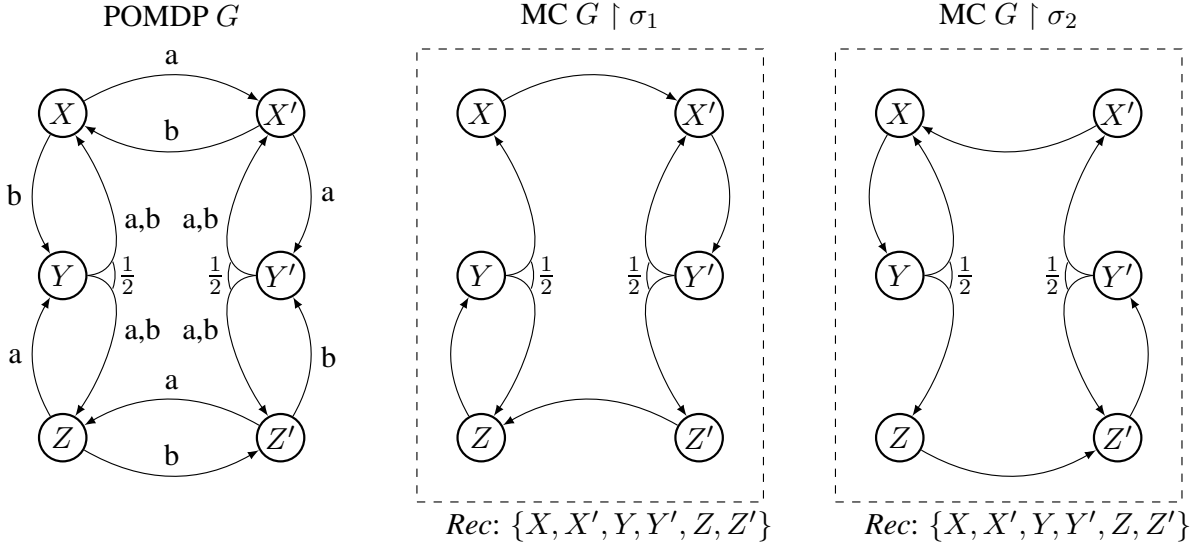
\begin{figure}
\begin{center}
\resizebox{\linewidth}{!}{
\begin{tikzpicture}[>=latex]
	\node[State,text=] (x) {$X$};
	\node[State,right of=x,xshift=50] (x') {$X'$};
	\node[State,below of=x,yshift=-30] (y) {$Y$}; 
	\node[State,below of=x',yshift=-30] (y') {$Y'$}; 
	\node[State,below of=y,yshift=-30] (z) {$Z$}; 
	\node[State,below of=y',yshift=-30] (z') {$Z'$}; 
	\draw[->]
	(x) edge[bend left] node[above] {a} (x')
	(x) edge[bend right] node[left] {b} (y)
	(x') edge[bend left] node[above] {b} (x)
	(x') edge[bend left] node[right] {a} (y')
	(z) edge[bend left] node[left] {a} (y)
	(z) edge[bend right] node[above] {b} (z')
	(z') edge[bend right] node[above] {a} (z)
	(z') edge[bend right] node[right] {b} (y')
	(y) edge[bend right,out=270] node[right] {a,b} (x)
	(y') edge[bend left,out=90] node[left] {a,b} (x')
	(y) edge[bend left,out=90] node[right] {a,b} (z)
	(y') edge[bend right,out=270] node[left] {a,b} (z')
	;
	\node[right of=y,xshift=-5] {$\frac{1}{2}$};
	\draw ($(y) +(18pt, 5pt)$) arc (20:-20:0.5cm);
	\draw ($(y') +(-18pt, -5pt)$) arc (-160:-200:0.5cm);
	\node[left of=y',xshift=5] {$\frac{1}{2}$};
	\node[fit=(x) (x') (y) (y') (z) (z'),inner sep=0.8cm] (mdp) {};
	\node[above of=mdp,yshift=65] {POMDP $G$};
	\node[State,right of=x',xshift=4em] (x1) {$X$};
	\node[State,right of=x1,xshift=50] (x1') {$X'$};
	\node[State,below of=x1,yshift=-30] (y1) {$Y$}; 
	\node[State,below of=x1',yshift=-30] (y1') {$Y'$}; 
	\node[State,below of=y1,yshift=-30] (z1) {$Z$}; 
	\node[State,below of=y1',yshift=-30] (z1') {$Z'$}; 
	\draw[->]
	(x1) edge[bend left] node[above] {} (x1')
	(x1') edge[bend left] node[right] {} (y1')
	(z1) edge[bend left] node[left] {} (y1)
	(z1') edge[bend right] node[above] {} (z1)
	(y1) edge[bend right,out=270] node[right] {} (x1)
	(y1') edge[bend left,out=90] node[left] {} (x1')
	(y1) edge[bend left,out=90] node[right] {} (z1)
	(y1') edge[bend right,out=270] node[left] {} (z1')
	;
	
	\draw ($(y1) +(18pt, 5pt)$) arc (20:-20:0.5cm);
	\draw ($(y1') +(-18pt, -5pt)$) arc (-160:-200:0.5cm);
	\node[right of=y1,xshift=-5] {$\frac{1}{2}$};
	\node[left of=y1',xshift=5] {$\frac{1}{2}$};
	\node[rectangle,fit= (x1) (x1') (y1) (y1') (z1) (z1'),inner sep=0.5cm,draw=black, dashed] (rec1) {};
	\node[above of=rec1,yshift=65] {MC $G\restr{\sigma_1}$};
	\node[below of=rec1,yshift=-62,xshift=7] {\emph{Rec}: $\{X,X',Y,Y',Z,Z'\}$};

	\node[State,right of=x1',xshift=4em] (x2) {$X$};
	\node[State,right of=x2,xshift=50] (x2') {$X'$};
	\node[State,below of=x2,yshift=-30] (y2) {$Y$}; 
	\node[State,below of=x2',yshift=-30] (y2') {$Y'$}; 
	\node[State,below of=y2,yshift=-30] (z2) {$Z$}; 
	\node[State,below of=y2',yshift=-30] (z2') {$Z'$}; 
	\draw[->]
	(x2) edge[bend right] node[left] {} (y2)
	(x2') edge[bend left] node[above] {} (x2)
	(z2) edge[bend right] node[above] {} (z2')
	(z2') edge[bend right] node[right] {} (y2')
	(y2) edge[bend right,out=270] node[right] {} (x2)
	(y2') edge[bend left,out=90] node[left] {} (x2')
	(y2) edge[bend left,out=90] node[right] {} (z2)
	(y2') edge[bend right,out=270] node[left] {} (z2')
	;
	\draw ($(y2) +(18pt, 5pt)$) arc (20:-20:0.5cm);
	\draw ($(y2') +(-18pt, -5pt)$) arc (-160:-200:0.5cm);
	\node[right of=y2,xshift=-5] {$\frac{1}{2}$};
	\node[left of=y2',xshift=5] {$\frac{1}{2}$};
	\node[rectangle,fit= (x2) (x2') (y2) (y2') (z2) (z2') ,inner sep=0.5cm,draw=black, dashed] (rec2) {};
	\node[below of=rec2,yshift=-62,xshift=7] {\emph{Rec}: $\{X,X',Y,Y',Z,Z'\}$};
	\node[above of=rec2,yshift=65] {MC $G\restr{\sigma_2}$};

\end{tikzpicture}
}
\end{center}
\caption{Belief is not sufficient}
\label{fig:simple_example}
\end{figure}

\begin{figure}[h]
\centering
\resizebox{7cm}{!}{
\begin{tikzpicture}[>=latex]
	\node[BState,text=] (xa) {$Xa$};
	\node[BState,right of=xa,xshift=50] (x'b) {$X'b$};
	\node[BState,below of=xa,yshift=-30] (yb) {$Yb$};
	\node[BState,below of=x'b,yshift=-30] (y'a) {$Y'a$};
	\node[BState,below of=yb,yshift=-30] (za) {$Za$};
	\node[BState,below of=y'a,yshift=-30] (z'b) {$Z'b$};
	\draw[->]
	(xa) edge[bend left] node [above] {} (x'b)
	(x'b) edge[bend left] node [above] {} (xa)
	(yb) edge[] node [above] {} (xa)
	(y'a) edge[] node [above] {} (x'b)
	(za) edge[bend left] node [above] {} (yb)
	(yb) edge[bend left] node [above] {} (za)
	(y'a) edge[bend left] node [above] {} (z'b)
	(z'b) edge[bend left] node [above] {} (y'a)
	;
	\node[rectangle,fit= (xa) (x'b) ,inner sep=0.3cm,draw=black, dashed] (rec1) {};
	\node[above of=rec1,xshift=27] {\emph{Rec:} $\{Xa,X'b\}$};
	
	\node[BState,right of=z'b,xshift=50] (zb) {$Zb$};
	\node[BState,right of=zb, xshift=50] (z'a) {$Z'a$};
	\node[BState,above of=z'a,yshift=30] (y'b) {$Y'b$};
	\node[BState,above of=zb,yshift=30] (ya) {$Ya$};
	\node[BState,above of=y'b,yshift=30] (x'a) {$X'a$};
	\node[BState,above of=ya,yshift=30] (xb) {$Xb$};
	\draw[->]
	(z'a) edge[bend left] node [above] {} (zb)
	(zb) edge[bend left] node [above] {} (z'a)
	(y'b) edge[] node [above] {} (z'a)
	(ya) edge[] node [above] {} (zb)
	(y'b) edge[bend left] node [above] {} (x'a)
	(x'a) edge[bend left] node [above] {} (y'b)
	(ya) edge[bend left] node [above] {} (xb)
	(xb) edge[bend left] node [above] {} (ya)
	;
	\node[rectangle,fit= (z'a) (zb) ,inner sep=0.3cm,draw=black, dashed] (rec2) {};
	\node[above of=rec2,yshift=115,xshift=27] {\emph{Rec:} $\{Zb,Z'a\}$};

\end{tikzpicture}
}
\caption{The Markov chain $G \restr \sigma_4$.}
\label{fig:simple_example1}
\end{figure}
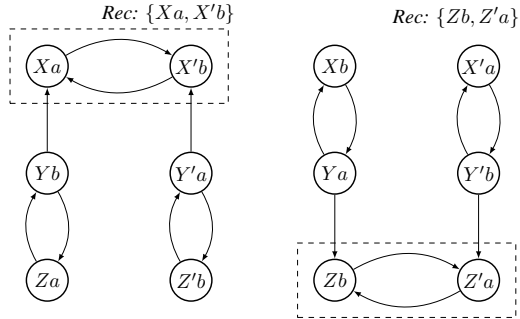

\begin{example}\label{ex:belief-suff}
We consider a POMDP with state space $\set{s_0,X,X',Y,Y',Z,Z'}$ and
action set $\set{a,b}$, and let $U= \set{X,X',Y,Y',Z,Z'}$.
From the initial state $s_0$ all the other states are reached with 
uniform probability in one-step, i.e., for all $s' \in U=\set{X,X',Y,Y',Z,Z'}$ 
we have $\trans(s_0,a)(s')=\trans(s_0,b)(s')=\frac{1}{6}$.
The transitions from the other states are as follows 
(shown in Figure~\ref{fig:simple_example}): 
(i)~$\trans(X,a)(X')=1$ and $\trans(X,b)(Y)=1$;
(ii)~$\trans(X',a)(Y')=1$ and $\trans(X',b)(X)=1$;
(iii)~$\trans(Z,a)(Y)=1$ and $\trans(Z,b)(Z')=1$;
(iv)~$\trans(Z',a)(Z)=1$ and $\trans(Z',b)(Y')=1$;
(v)~$\trans(Y,a)(X)=\trans(Y,b)(X)=\trans(Y,a)(Z)=\trans(Y,b)(Z)=\frac{1}{2}$;
and
(vi)~$\trans(Y',a)(X')=\trans(Y',b)(X')=\trans(Y',a)(Z')=\trans(Y',b)(Z')=\frac{1}{2}$.
All states in $U$ have the same observation.
The coB\"uchi objective is given by the target set 
$\set{X,X',Z,Z'}$, i.e., $Y$ and $Y'$ must be visited only finitely often.

The belief initially after one-step is the set $U$ since
from $s_0$ all of them are reached with positive probability.
The belief is always the set $U$ since every state has an input edge for 
every action, i.e., if the current belief is $U$ (i.e., the set of states 
that the POMDP is currently in with positive probability is $U$), then
irrespective of whether $a$ or $b$ is chosen all states of $U$ are reached
with positive probability and hence the belief set is again $U$.
There are three belief-based stationary strategies: (i)~$\sigma_1$ that plays always $a$;
(ii)~$\sigma_2$ that plays always $b$; or (iii)~$\sigma_3$ that plays 
both $a$ and $b$ with positive probability.
The Markov chains $G \restr \sigma_1$ (resp. $\sigma_2$ and $\sigma_3$) 
are obtained by retaining the edges labeled by action $a$ (resp. action $b$, 
and both actions $a$ and $b$).
For all the three strategies, the Markov chains obtained have the whole 
set $U$ as the recurrent class, and hence both $Y$ and $Y'$ 
are visited infinitely often with probability~1 violating the 
coB\"uchi objective.
The Markov chains $G \restr \sigma_1$ and $G \restr \sigma_2$ are also 
shown in Figure~\ref{fig:simple_example}, and the graph of 
$G \restr \sigma_3$ is the same as the POMDP (with edge labels removed).
The strategy $\sigma_4$ that plays action $a$ and $b$ alternately gives 
rise to the Markov chain $G \restr \sigma_4$ (shown in 
Figure~\ref{fig:simple_example1}) 
(i.e., $\sigma_4$ has two memory states $a$ and $b$, in memory state 
$a$ it plays action $a$ and switches to memory state $b$, and in memory state 
$b$ it plays action $b$ and switches to memory state $a$).
The recurrent classes do not intersect with $(Y,m)$ or $(Y',m)$, for 
memory state $m \in \set{a,b}$, and hence is a finite-memory almost-sure winning strategy.
\qed
\end{example}





%

In Example~\ref{ex:belief-suff} the coB\"uchi objective is not a visible objective.
In the following example we modify Example~\ref{ex:belief-suff} to show that 
 belief-based stationary strategies are
not sufficient even if we consider visible coB\"uchi objectives.

\begin{example}\label{ex:complex}
We consider the POMDP shown in Figure~\ref{fig:complex_example}:
the transition edges in the set $U=\set{X,X',Y,Y',Z,Z'}$ are exactly
the same as in Figure~\ref{fig:simple_example}, and the transition 
probabilities are always uniform over the support set. 
We add a new state $B$ and from the state $Y$ and $Y'$ add positive 
transition probabilities (probability~$\frac{1}{3}$) to the state
$B$ for both actions $a$ and $b$.
Recall that $Y$ and $Y'$ were the bad states in Example~\ref{ex:belief-suff}.
From state $B$ all states in $U$ are reached with positive probability for both actions
$a$ and $b$.
All states in $U$ have the same observation (denoted as $o_U$), 
and the state $B$ has a new and different observation (denoted as $o_B$).
The coB\"uchi objective is to visit only states with observation $o_U$
infinitely often (i.e., to avoid to visit state $B$ infinitely often).
Note that the objective is a visible objective. 
Since we retain all edges as in Figure~\ref{fig:simple_example} and
from $B$ all states in $U$ are reached with positive probability
in one step, whenever the current observation is $o_U$, then the belief
is the set $U$.
As in Example~\ref{ex:belief-suff} there are three belief-based stationary
strategies ($\straa_1,\straa_2$ and $\straa_3$) in belief $U$, 
and the Markov chains under $\straa_1$ and $\straa_2$ are shown
in Figure~\ref{fig:complex_example}, and the Markov chain under
$\straa_3$ has the same edges as the original POMDP.
For all the belief-based stationary strategies the recurrent class
contains the state $B$, and hence $B$ is visited infinitely often
with probability~1 violating the coB\"uchi objective.
The strategy $\straa_4$ that alternates actions $a$ and $b$ 
is a finite-memory almost-sure winning strategy and the Markov chain 
obtained given $\straa_4$ is shown in Figure~\ref{fig:complex_example1}.
Also note that our example shows that belief-based stationary strategies
are also not sufficient for positive winning for 
coB\"uchi objectives.
\end{example}

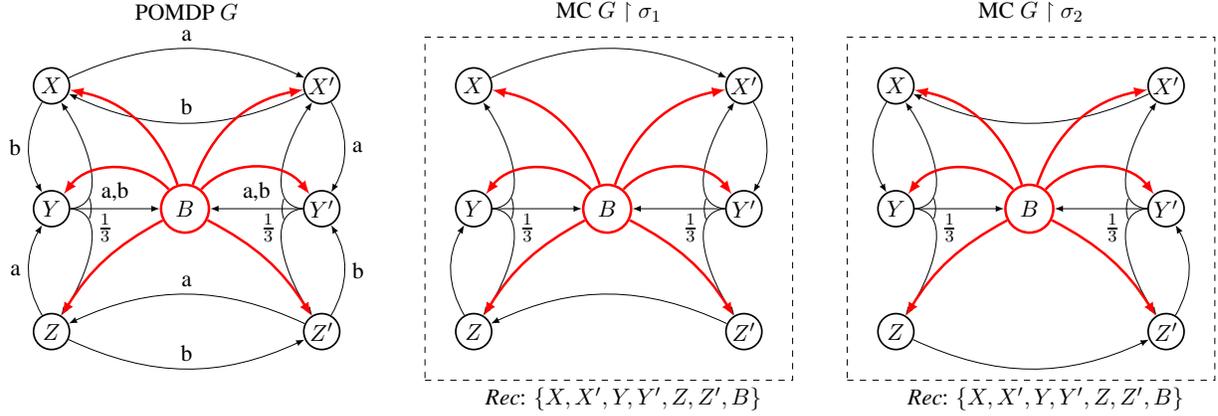
\begin{figure}
\begin{center}
\resizebox{\linewidth}{!}{
\begin{tikzpicture}[>=latex]
	\node[State,text=] (x) {$X$};
	\node[State,right of=x,xshift=100] (x') {$X'$};
	\node[State,below of=x,yshift=-30] (y) {$Y$}; 
	\node[State,below of=x',yshift=-30] (y') {$Y'$}; 
	\node[State,below of=y,yshift=-30] (z) {$Z$}; 
	\node[State,below of=y',yshift=-30] (z') {$Z'$}; 
	\node[RState,right of=y,xshift=35] (b) {$B$};
	\draw[->]
	(x) edge[bend left=25] node[above] {a} (x')
	(x) edge[bend right] node[left] {b} (y)
	(x') edge[bend left=25] node[above] {b} (x)
	(x') edge[bend left] node[right] {a} (y')
	(z) edge[bend left] node[left] {a} (y)
	(z) edge[bend right=25] node[above] {b} (z')
	(z') edge[bend right=25] node[above] {a} (z)
	(z') edge[bend right] node[right] {b} (y')
	(y) edge[bend right,out=270]   (x)
	(y) edge[] node[above] {a,b} (b)
	(y') edge[] node[above] {a,b} (b)
	(y') edge[bend left,out=90] (x')
	(y) edge[bend left,out=90]  (z)
	(y') edge[bend right,out=270] (z')
	(b) edge[red, very thick, bend left=50] (y')
	(b) edge[red, very thick, bend right=50] (y)
	(b) edge[red, very thick, bend right] (x)
	(b) edge[red, very thick, bend left] (x')
	(b) edge[red, very thick, bend right=15] (z)
	(b) edge[red, very thick, bend left=15] (z')
	;
	\node[right of=y,xshift=-3,yshift=-9] {$\frac{1}{3}$};
	\draw ($(y) +(18pt, 5pt)$) arc (20:-20:0.5cm);
	\draw ($(y') +(-18pt, -5pt)$) arc (-160:-200:0.5cm);
	\node[left of=y',xshift=3,yshift=-9] {$\frac{1}{3}$};
	\node[fit=(x) (x') (y) (y') (z) (z'),inner sep=0.8cm] (mdp) {};
	\node[above of=mdp,yshift=65] {POMDP $G$};

	\node[State,right of =x',xshift=4em] (x1) {$X$};
	\node[State,right of=x1,xshift=100] (x1') {$X'$};
	\node[State,below of=x1,yshift=-30] (y1) {$Y$}; 
	\node[State,below of=x1',yshift=-30] (y1') {$Y'$}; 
	\node[State,below of=y1,yshift=-30] (z1) {$Z$}; 
	\node[State,below of=y1',yshift=-30] (z1') {$Z'$}; 
	\node[RState,right of=y1,xshift=35] (b1) {$B$};
	\draw[->]
	(x1) edge[bend left=25] node[above] {} (x1')
	(x1') edge[bend left] node[right] {} (y1')
	(z1) edge[bend left] node[left] {} (y1)
	(z1') edge[bend right=25] node[above] {} (z1)
	(y1) edge[bend right,out=270]   (x1)
	(y1') edge[bend left,out=90] (x1')
	(y1) edge[bend left,out=90]  (z1)
	(y1') edge[bend right,out=270] (z1')
	(b1) edge[red, very thick, bend left=50] (y1')
	(b1) edge[red, very thick, bend right=50] (y1)
	(b1) edge[red, very thick, bend right] (x1)
	(b1) edge[red, very thick, bend left] (x1')
	(b1) edge[red, very thick, bend right=15] (z1)
	(b1) edge[red, very thick, bend left=15] (z1')
	(y1) edge[] node[above] {} (b1)
	(y1') edge[] node[above] {} (b1)
	;
	\node[right of=y1,xshift=-3,yshift=-9] {$\frac{1}{3}$};
	\draw ($(y1) +(18pt, 5pt)$) arc (20:-20:0.5cm);
	\draw ($(y1') +(-18pt, -5pt)$) arc (-160:-200:0.5cm);
	\node[left of=y1',xshift=3,yshift=-9] {$\frac{1}{3}$};
	\node[rectangle,fit= (x1) (x1') (y1) (y1') (z1) (z1') ,inner sep=0.5cm,draw=black, dashed] (straa1) {};
	\node[below of=straa1,yshift=-62,xshift=7] {\emph{Rec}: $\{X,X',Y,Y',Z,Z',B\}$};
	\node[above of=straa1, yshift=65] {MC $G\restr{\sigma_1}$};

	\node[State,right of =x1',xshift=4em] (x2) {$X$};
	\node[State,right of=x2,xshift=100] (x2') {$X'$};
	\node[State,below of=x2,yshift=-30] (y2) {$Y$}; 
	\node[State,below of=x2',yshift=-30] (y2') {$Y'$}; 
	\node[State,below of=y2,yshift=-30] (z2) {$Z$}; 
	\node[State,below of=y2',yshift=-30] (z2') {$Z'$}; 
	\node[RState,right of=y2,xshift=35] (b2) {$B$};
	\draw[->]
	(x2) edge[bend right] node[left] {} (y2)
	(x2') edge[bend left=25] node[above] {} (x2)
	(z2) edge[bend right=25] node[above] {} (z2')
	(z2') edge[bend right] node[right] {} (y2')
	(y2) edge[bend right,out=270]   (x2)
	(y2) edge[] node[above] {} (b2)
	(y2') edge[] node[above] {} (b2)
	(y2') edge[bend left,out=90] (x2')
	(y2) edge[bend left,out=90]  (z2)
	(y2') edge[bend right,out=270] (z2')
	(b2) edge[red, very thick, bend left=50] (y2')
	(b2) edge[red, very thick, bend right=50] (y2)
	(b2) edge[red, very thick, bend right] (x2)
	(b2) edge[red, very thick, bend left] (x2')
	(b2) edge[red, very thick, bend right=15] (z2)
	(b2) edge[red, very thick, bend left=15] (z2')

	;
	\node[right of=y2,xshift=-3,yshift=-9] {$\frac{1}{3}$};
	\draw ($(y2) +(18pt, 5pt)$) arc (20:-20:0.5cm);
	\draw ($(y2') +(-18pt, -5pt)$) arc (-160:-200:0.5cm);
	\node[left of=y2',xshift=3,yshift=-9] {$\frac{1}{3}$};
	\node[rectangle,fit= (x2) (x2') (y2) (y2') (z2) (z2') ,inner sep=0.5cm,draw=black, dashed] (straa2) {};
	\node[below of=straa2,yshift=-62,xshift=7] {\emph{Rec}: $\{X,X',Y,Y',Z,Z',B\}$};
	\node[above of=straa2, yshift=65] {MC $G\restr{\sigma_2}$};
\end{tikzpicture}
}
\caption{Belief is not sufficient}
\label{fig:complex_example}
\end{center}
\end{figure}

\begin{figure}[h]
\centering

\resizebox{10cm}{!}{
\begin{tikzpicture}[>=latex]

	\node[BState,text=] (za) {$Za$};
	\node[BState,right of=za,xshift=50] (yb) {$Yb$};
	\node[BState,right of=yb,xshift=50] (xa) {$Xa$};
	\node[BState,right of=xa,xshift=50] (x'b) {$X'b$};
	\node[BState,right of=x'b,xshift=50] (y'a) {$Y'a$};
	\node[BState,right of=y'a,xshift=50] (z'b) {$Z'b$};
	\draw[->]
	(xa) edge[bend left] node [above] {} (x'b)
	(x'b) edge[bend left] node [above] {} (xa)
	(yb) edge[] node [above] {} (xa)
	(y'a) edge[] node [above] {} (x'b)
	(za) edge[bend left] node [above] {} (yb)
	(yb) edge[bend left] node [above] {} (za)
	(y'a) edge[bend left] node [above] {} (z'b)
	(z'b) edge[bend left] node [above] {} (y'a)
	;
	\node[rectangle,fit= (xa) (x'b) ,inner sep=0.3cm,draw=black, dashed] (rec1) {};
	\node[above of=rec1,xshift=27] {\emph{Rec:} $\{Xa,X'b\}$};

	\node[BState,below of=yb,yshift=-30] (ba) {$Ba$};
	\node[BState,below of=y'a,yshift=-30] (bb) {$Bb$};
	\draw[->]
	(yb) edge[]  (ba)
	(y'a) edge[]  (bb)
	(ba) edge[bend left] (yb)
	(ba) edge[bend right=10] (x'b)
	(ba) edge[bend right=10] (z'b)
	(bb) edge[bend right] (y'a)
	(bb) edge[bend left=10] (xa)
	(bb) edge[bend left=10] (za)
	;
	\node[BState,below of=za,yshift=-100] (x'a) {$X'a$};
	\node[BState,right of=x'a,xshift=50] (y'b) {$Y'b$};
	\node[BState,right of=y'b,xshift=50] (z'a) {$Z'a$};
	\node[BState,right of=z'a,xshift=50] (zb) {$Zb$};
	\node[BState,right of=zb,xshift=50] (ya) {$Ya$};
	\node[BState,right of=ya,xshift=50] (xb) {$Xb$};
	\draw[->]
	(z'a) edge[bend left] node [above] {} (zb)
	(zb) edge[bend left] node [above] {} (z'a)
	(y'b) edge[] node [above] {} (z'a)
	(ya) edge[] node [above] {} (zb)
	(y'b) edge[bend left] node [above] {} (x'a)
	(x'a) edge[bend left] node [above] {} (y'b)
	(ya) edge[bend left] node [above] {} (xb)
	(xb) edge[bend left] node [above] {} (ya)
	;
	\node[rectangle,fit= (z'a) (zb) ,inner sep=0.3cm,draw=black, dashed] (rec2) {};
	\node[below of=rec2,xshift=27] {\emph{Rec:} $\{Z'a,Zb\}$};
	\draw[->]
	(y'b) edge[]  (ba)
	(ya) edge[]  (bb)
	(ba) edge[bend right] (y'b)
	(ba) edge[bend left=10] (xb)
	(ba) edge[bend left=10] (zb)
	(bb) edge[bend left] (ya)
	(bb) edge[bend right=10] (x'a)
	(bb) edge[bend right=10] (z'a)
	;

\end{tikzpicture}
}
\caption{The Markov chain $G \restr \sigma_4$.}
\label{fig:complex_example1}
\end{figure}
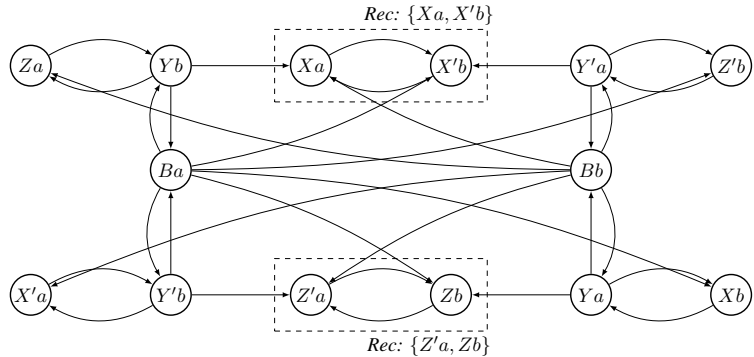

\begin{remark}
In Example~\ref{ex:complex} we have shown that 
belief-based stationary strategies are not sufficient for finite-memory
almost-sure and positive winning strategies in POMDPs
with coB\"uchi objectives.
In contrast, for almost-sure winning for B\"uchi objectives
in POMDPs, belief-based stationary strategies are 
sufficient~\cite{CDHR06} (both for finite-memory and infinite-memory 
strategies).
The fact that belief-based stationary strategies are not 
sufficient for finite-memory positive winning strategies in 
POMDPs with B\"uchi objectives can be obtained from a 
simple modification of Example~\ref{ex:complex} as follows:
we consider the POMDP in Example~\ref{ex:complex} and change
the state $B$ to an absorbing state. 
The B\"uchi objective is to visit the observation $o_U$ 
infinitely often, and for all the three belief-based stationary strategies
$\straa_1, \straa_2$, and $\straa_3$ the Markov chain has only
one recurrent class consisting of the absorbing state $B$.
The strategy $\straa_4$ ensures that with positive probability
a recurrent class is contained in $o_U$ and is a finite-memory 
positive winning strategy.
Finally, for positive winning in POMDPs with coB\"uchi objectives,
the EXPTIME-complete computational complexity result was obtained
with the following straight forward observation~\cite{CDH10a}: given a POMDP $G$
with a coB\"uchi objective $\coBuchi(\target)$, let $S_W$ be the set 
of states $s$ such that if $s$ is the starting state (i.e., 
initial belief is $\set{s}$), then almost-sure safety can be 
ensured for the target set (i.e., $\Safe(\target)$ can be ensured
almost-surely).
Then positive winning for coB\"uchi coincides with positive reachability
to the set $S_W$ because as soon as $S_W$ is reached, then the current 
belief contains a state in $S_W$, and then with positive probability the 
strategy can assume that the current state is a state in $S_W$ and play 
the almost-sure safety strategy and the strategy ensures that the coB\"uchi
objective is satisfied with positive probability.
Conversely it was also shown that a positive winning strategy for the 
coB\"uchi objective must ensure positive probability reachability to 
$S_W$~\cite{CDH10a}.
Hence positive winning for coB\"uchi objectives can be ensured by solving
almost-sure safety and positive reachability, and thus we obtain the
EXPTIME-complete result from results of almost-sure safety and positive
reachability.
However, from the previous construction it was not clear whether belief-based stationary 
strategies are sufficient or not, and Example~\ref{ex:complex} shows
that belief-based stationary strategies are not sufficient. 
\end{remark}





\subsection{Upper bound on memory of finite-memory strategies}
For the following of the section, we fix a POMDP $\game=(S,\Act,\trans,\Obs,\obsmap,s_0)$,
with a Muller objective $\Muller(\calf)$ with the set $\priority$ of colors and
a color mapping function $\col$. 
We will denote by $\D$ the powerset of the powerset of the set $\priority$  
of colors, 
i.e., $\D=\powset(\powset(\priority))$; and
note that $|\D|=2^{2^{d}}$, where $d=|\priority|$.
The goal of the section is to prove the following fact: given a
finite-memory almost-sure (resp. positive) winning strategy $\straa$ 
on $\game$ there exists a finite-memory almost-sure (resp. positive) 
winning strategy $\straa'$ on $\game$, of memory size at most 
$\Mem^*=2^{|S|} \cdot 2^{|S|} \cdot |\D|^{|S|} $.

\smallskip\noindent{\bf Overview of the proof.} We first present an overview of 
our proof structure.
\begin{itemize}
\item Given an arbitrary finite-memory strategy $\sigma$ we will consider the 
Markov chain $G \restr \sigma$ arising by fixing the strategy.

\item Given the Markov chain we will define a projection graph that depends 
on the recurrent classes of the Markov chain. 
The projection graph is of size at most $\Mem^*$. 

\item Given the projection graph we will construct a projected strategy with 
memory size at most $\Mem^*$ that preserves the recurrent classes of the 
Markov chain $G \restr \sigma$.
\end{itemize}

\smallskip\noindent{\bf Notations.}
Given $Z\in\D^{|S|}$ and given $s\in S$, we write $Z(s)$ 
(which is in $\D=\powset(\powset(D))$) 
for the $s$-component of $Z$. 
For two sets $U_1$ and $U_2$ and $U \subseteq U_1\times U_2$, we denote by $\proj_1(U)$ 
the projection of $U$ on the first component, formally, 
$\proj_1(U)=\set{u_1 \in U_1 \mid \exists u_2 \in U_2. (u_1,u_2) \in U}$;
and the definition of $\proj_2(U)$ for the projection on the second component 
is analogous.

\smallskip\noindent{\bf Basic definitions for the projection graph.}
We now introduce notions associated with the finite Markov chain 
$G \restr \sigma$ that will be essential in defining the projection graph.

\begin{definition}[Recurrence set functions]
Let $\straa$ be a finite-memory strategy with memory $M$ on $\game$ for the
Muller objective with the set $\priority$ of colors, and let $m \in M$.
\begin{itemize}
 \item \emph{(Function set recurrence).} 
The function $\SetRec_{\straa}(m): S \to \D$  maps every state 
$s \in S$ to the projections of colors 
of recurrent classes reachable from 
$(s,m)$ in $\game \restr \sigma$.
Formally, $\SetRec_{\straa}(m)(s) =
\set{\col(\proj_1(U)) \mid U \in \Rec(\game \restr \straa)((s,m))}$,
i.e., we consider the set $\Rec(\game \restr \straa)((s,m))$ of 
recurrent classes reachable from the state $(s,m)$ in $\game \restr \straa$, 
obtain the projections on the state space $S$ and consider the colors 
of states in the projected set.
We will in sequel consider $\SetRec_{\straa}(m) \in \D^{|S|}$.


\item \emph{(Function boolean recurrence).} The function  $\BoolRec_\straa(m): S \to \set{0,1}$ 
 is such that for all $s\in S$, we have $\BoolRec_\straa(m)(s)=1$ if there exists $U\in \Rec(\game \restr \straa)((s,m))$ 
such that  $(s,m)\in U$, and $0$ if not.
Intuitively, $\BoolRec_{\straa}(m)(s)=1$ if $(s,m)$ belongs to a recurrent class in $\game \restr \straa$ and
$0$ otherwise.
In sequel we will consider $\BoolRec_{\straa}(m) \in \set{0,1}^{|S|}$.
\end{itemize}
\end{definition}

%

We first define the projection graph and then present  
a simple property of $\SetRec_{\straa}(m)$ function 
related to the reachability property.

\begin{definition}[Projection graph]
Let $\straa$ be a finite-memory strategy. 
We define the \emph{projection graph} $\PG(\straa)=(V,E)$ associated to 
$\straa$ as follows:
\begin{itemize}
\item \emph{(Vertex set).} 
The set of vertices is $V=\set{ (U,\BoolRec_\straa(m),\SetRec_\straa(m)) 
\mid U\subseteq S\ \mathrm{and}\ m\in M}$. 
 
\item \emph{(Edge labels).} 
The edges of the graph are labeled by actions in $\Act$.

\item \emph{(Edge set).} 
Let $U\subseteq S$, $m\in M$ and $a\in \Supp(\straa_n(m))$. 
Let $\ov{U}=\bigcup_{s \in U} \Supp(\trans(s,a))$
denote the set of possible successors of states in $U$ given action $a$.
We add the following set of edges in $E$: Given $(U',m')$ such that 
there exists $o\in \Obs$ with $\obsmap^{-1}(o) \cap \ov{U}=U'$ and 
 $m'\in \Supp(\straa_u(m,o,a))$,
we add the edge $(U,\BoolRec_\straa(m),\SetRec_\straa(m))\stackrel{a}{\rightarrow}(U',\BoolRec_\straa(m'),\SetRec_\straa(m'))$ to $E$.
Intuitively, the update from $U$ to $U'$ is the update of the belief, i.e.,
if the previous belief is the set $U$ of states, and the current observation 
is $o$, then the new belief is $U'$; 
the update of $m$ to $m'$ is according to the support of the memory update
function; and the $\BoolRec$ and $\SetRec$ functions for the memories are
given by the strategy $\straa$.
 
\item \emph{(Initial vertex).} 
The \emph{initial vertex} of $\PG(\straa)$ is the vertex $(\set{s_0},\BoolRec_\straa(m_0),\SetRec_\straa(m_0))$.
\end{itemize}

\end{definition}
Note that $V\subseteq\powset(S)\times \set{0,1}^{|S|} \times 
\D^{|S|}$, and hence $|V|\leq \Mem^*$. For the rest of this section we fix an arbitrary finite-memory strategy $\straa$ that uses memory $M$.
\begin{lemma}\label{l-incl-rec}
Let $s,s'\in S$ and $m,m'\in M$ be such that $(s',m')$ is reachable from 
$(s,m)$ in $\game \restr \straa$. 
Then $\SetRec_\straa(m')(s') \subseteq \SetRec_\straa(m)(s)$.
\end{lemma}
\begin{proof}
Since $(s',m')$ is reachable from $(s,m)$ in $\game \restr \straa$, it follows 
by simple properties of Markov chains and recurrent classes that we have 
$\Rec(\game \restr \straa)((s',m')) \subseteq \Rec(\game \restr \straa)((s,m))$ 
(Property~4 of Markov chains).
The $\subseteq$ relation is preserved under the projections on states, and 
then considering the color mapping. 
Hence the result follows.
\end{proof}

In the following two lemmas we establish the connection of functions 
$\BoolRec_{\straa}(m)$ and $\SetRec_{\straa}(m)$ with the edges of 
the projection graph. 
The intuitive description of the first lemma is as follows: 
it shows that if $\BoolRec$ is set to~1 for a vertex of the projection graph, 
then for all successors according to the edges of the projection graph,  
$\BoolRec$ is also~1 for the successors.
The second lemma shows a similar result for the projection graph showing 
that the $\SetRec$ functions are subsets for each component for 
successor vertices.

\begin{lemma}\label{l-C-croit}
Let $(V,E)=\PG(\straa)$ be the projection graph of $\straa$.
Let $(U,B,L)\stackrel{a}{\rightarrow}(U',B',L')$ be an edge in $E$, 
where $U,U'\subseteq S$, $B,B'\in\set{0,1}^{|S|}$, and 
$L,L'\in\D^{|S|}$. 
Then for all $s\in U$ and $s'\in \Supp(\trans(s,a))$ the following assertion 
holds:  if $B(s) =1$, then $B'(s')=1$.
\end{lemma}

\begin{proof}
We first note that if $(U,B,L) \stackrel{a}{\rightarrow}(U',B',L')$ is an 
edge in $E$, then there exists memory state $m$ and $m'$ such that 
(i)~$B=\BoolRec_{\straa}(m)$ and $L=\SetRec_{\straa}(m)$;
(ii)~$B'=\BoolRec_{\straa}(m')$ and $L'=\SetRec_{\straa}(m')$;
(iii)~$a\in \Supp(\straa_n(m))$ and $m'\in \Supp(\straa_u(m,\obsmap(s'),a))$.
Hence it follows that $(s',m')$ is reachable in one step from $(s,m)$ in 
$\game \restr \straa$.
Now, if $(s',m')$ is reached with positive probability from $(s,m)$ in 
$\game \restr \straa$ and if $(s,m)$ is a recurrent state of $\game \restr
\straa$, then $(s',m')$ is also recurrent and lies in the same recurrent 
class as $(s,m)$ (since both $(s,m)$ and $(s',m')$ would lie in the same 
bottom scc of the graph of the Markov chain).
Thus if $\BoolRec_{\straa}(m)(s)=1$, then 
$\BoolRec_{\straa}(m')(s')=1$.
Since $B=\BoolRec_\straa(m)$ and $B'=\BoolRec_\straa(m')$, 
the desired result follows.
\end{proof}

\begin{lemma}\label{l-Z-decrease}
Let $(V,E)=\PG(\straa)$ be the projection graph of $\straa$.
Let $(U,B,L)\stackrel{a}{\rightarrow}(U',B',L')$ be an edge in $E$, 
where $U,U'\subseteq S$, $B,B'\in\set{0,1}^{|S|}$, and 
$L,L'\in\D^{|S|}$. 
Then for all $s\in S$ and all $s'\in \Supp(\trans(s,a))$, we have 
$L'(s')\subseteq L(s)$.
\end{lemma}
\begin{proof}
By definition of $\PG(\straa)$ (as in the proof of Lemma~\ref{l-C-croit}), 
there exist $m,m'\in M$ such that (i)~$B=\BoolRec_{\straa}(m)$ and $L=\SetRec_{\straa}(m)$;
(ii)~$B'=\BoolRec_{\straa}(m')$ and $L'=\SetRec_{\straa}(m')$.
Moreover, $a\in \Supp(\straa_n(m))$, and $(U',m')$ is such that there exists
$o \in \Obs$ with $U' =(\bigcup_{s\in U} \Supp(\trans(s,a)) \cap \obsmap^{-1}(o) $
and $m' \in \Supp(\straa_u(m,o,a))$.
This implies that $(s',m')$ is reachable from $(s,m)$ in $\game \restr \straa$.
As a consequence
$\Rec(\game \restr \straa)((s',m')) \subseteq 
\Rec(\game \restr \straa)((s,m))$ (Property~4), and as $\subseteq$ relation is preserved by the projection on the states and then on the colors, 
it follows that $L'(s')\subseteq L(s)$.
\end{proof}

We now define the notion of projected strategies: intuitively
the projected strategy collapses memory with same $\BoolRec$ and $\SetRec$
functions, and at a collapsed memory state plays uniformly the union of the 
actions played at the corresponding memory states.

\begin{definition}[Projected strategy $\projp(\straa)$ of a finite-memory 
strategy]
Let $\PG(\straa)=(V,E)$ be the projection graph of $\straa$. 
We define the following projected strategy $\straa'=\projp(\straa)=(\straa_u',\straa_n',M',m_0')$:
\begin{itemize}
  \item \emph{(Memory set).} The memory set of $\projp(\straa)$ is $M'=V=\set{(U,\BoolRec_\straa(m),\SetRec_\straa(m)) \mid U\subseteq S\ \mathrm{and}\ m\in M}$.
  \item \emph{(Initial memory).} The initial memory state of $\projp(\straa)$ is
$m_0'=(\set{s_0},\BoolRec_\straa(m_0),\SetRec_\straa(m_0))$.
  \item \emph{(Memory update).} Let $m=(U,B,L)\in M'$, $o\in\Obs$ and $a\in\Act$. 
  Then $\straa_u'(m,o,a)$ is the uniform distribution over the set 
$\set{m'=(U',B',L')\in M'\mid  m\stackrel{a}{\rightarrow}m'\in E \mbox{ and } U'
\subseteq \obsmap^{-1}(o)}$.
  \item \emph{(Action selection).} Given $m\in M'$, the action selection function 
$\straa'_n(m)$ is the uniform distribution over 
$\set{a\in\Act\mid \exists m'\in M'\ s.t.\ m\stackrel{a}{\rightarrow}m'\in E}$.
\end{itemize}
\end{definition}

\smallskip\noindent{\bf Markov chain of the projected strategy.}
For the following of the section, we fix a finite-memory strategy $\straa$ on $\game$, 
 let $(V,E)=\PG(\straa)$ be the projection graph, and let $\straa'=\projp(\straa)$ be
the projected strategy. 
The finite-memory strategy $\straa'=(\straa_u',\straa_n',M',m_0')$ induces a probability transition function on 
$S\times M'$: given $s,s'\in S$ and $m,m'\in M'$ 
let $\trans_{\straa'}\big((s',m')\ |\ (s,m)\big)$ 
be the probability to go from state $(s,m)$ to state $(s',m')$ in one step if we use strategy $\straa'$. 
Formally, 
\[
\trans_{\straa'}\big((s',m')\ |\
(s,m)\big)= \sum_{a \in \Act} \straa_n'(m)(a) \cdot \trans(s,a)(s')\cdot\straa_u'(m,\obsmap(s'),a)(m').
\]
The chain $\game \restr \straa'$ is a finite state Markov chain, 
with state space $S\times M'$, which is a subset of $S\times \powset(S)\times \set{0,1}^{|S|}\times \D^{|S|}$. 
Given $X\in S$, $Y\subseteq S$, $C\in\set{0,1}^{|S|}$, and $Z\in\D^{|S|}$, 
let $\Succ_1((X,Y,C,Z))$ denote the set of states of the Markov chain reachable in one step from the 
state $(X,Y,C,Z)$.

\smallskip\noindent{\bf Random variable notations.} 
For all $n\geq0$ we write $X_n,Y_n,C_n,Z_n,W_n$ for the random variables which 
correspond respectively to the projection of the $n$-th state of the Markov
chain $\game \restr \straa'$ on the $S$ component, 
the $\powset(S)$ component, the $\set{0,1}^{|S|}$ component,
the $\D^{|S|}$ component, and the $n$-th action, respectively.

\smallskip\noindent{\bf Run of the Markov chain $\game \restr \straa'$.}
A \emph{run} on $\game \restr \straa'$ is a sequence 
\[
r=(X_0,Y_0,C_0,Z_0)\stackrel{W_0}{\rightarrow}(X_1,Y_1,C_1,Z_1)\stackrel{W_1}{\rightarrow}...
\]
such that each finite prefix of $r$ is generated with positive probability on the chain, 
i.e., for all $i \geq 0$, we have  
(i)~$W_i \in \Supp(\straa'_n(Y_i,C_i,Z_i))$; 
(ii)~$X_{i+1} \in \Supp(\trans(X_i,W_i))$; and
(iii)~$(Y_{i+1},C_{i+1},Z_{i+1}) \in 
\Supp(\straa'_u((Y_i,C_i,Z_i),\obsmap(X_{i+1}),W_i))$. 
In the following three lemmas we establish crucial properties of the Markov 
chain obtained from the projected strategy. 

\begin{lemma}
Let $X\in S$, $Y\subseteq S$, $C\in\set{0,1}^{|S|}$, and $Z\in\D^{|S|}$. 
Then:
\[
Z(X)=\bigcup_{(X',Y',C',Z')\in \Succ_1((X,Y,C,Z))} Z'(X').
\]
\end{lemma}
\begin{proof}
This follows from the following basic property of finite Markov chains: given a state $s$ of a finite Markov chain, the set of recurrent classes reachable from $s$ is the union of the 
set of recurrent classes reachable from the set of states reachable from $s$ in one step (Property~5 of Markov chains). The relation is preserved by projection on the colors 
of states.
\end{proof}

\begin{lemma}\label{l-decroit}
Let $\straa'=\projp(\straa)$ be 
the projected strategy of $\straa$.
Given $s,s'\in S$ and $m,m'\in M$, if $(s',m')$ is reachable from $(s,m)$ in $\game\restr \straa$, then for all $Y\subseteq S$ 
such that $(s,Y,\BoolRec_\straa(m),\SetRec_\straa(m))$ is a state of $\game \restr \straa'$, there exists $Y'\subseteq S$ 
such that $(s',Y',\BoolRec_\straa(m'),\SetRec_\straa(m'))$ is reachable from $(s,Y,\BoolRec_\straa(m),\SetRec_\straa(m))$ in $\game \restr \straa'$.
\end{lemma}
\begin{proof}
Suppose first that $(s',m')$ is reachable from $(s,m)$ in $\game\restr \straa$ in one step. 
Let $Y\subseteq S$ 
be such that $(s,Y,\BoolRec_\straa(m),\SetRec_\straa(m))$ is a state of $\game \restr \straa'$. 
Then there exists an edge in the projection graph of $\straa$ from 
$(Y,\BoolRec_\straa(m),\SetRec_\straa(m))$ to another vertex $(Y',\BoolRec_\straa(m'),\SetRec_\straa(m'))$. 
As a consequence, there exists $Y'\subseteq S$ 
such that $(s',Y',\BoolRec_\straa(m'),\SetRec_\straa(m'))$ is reachable from 
$(s,Y,\BoolRec_\straa(m),\SetRec_\straa(m))$ in $\game \restr \straa'$.

We conclude the proof by induction: if $(s',m')$ is reachable from $(s,m)$ in $\game\restr \straa$, then there exists a sequence 
of couples $(s_1,m_1),(s_2,m_2),...,(s_i,m_i)$ such that $(s_1,m_1)=(s,m)$,
$(s_i,m_i)=(s',m')$, and for all $j\in \set{1,...,i-1}$ 
we have that $(s_{j+1},m_{j+1})$ is reachable from $(s_{j},m_{j})$ in one step. 
Using the proof for an elementary step (or one step) inductively on such a sequence, we get 
the result.
\end{proof}

\begin{lemma}\label{l-chain_Z}
Let $X_0\in S$, $Y_0\in\powset(S)$, $C_0\in\set{0,1}^{|S|}$ and 
$Z_0\in\D^{|S|}$, and 
let $r=(X_0,Y_0,C_0,Z_0)\stackrel{W_0}{\rightarrow}(X_1,Y_1,C_1,Z_1)\stackrel{W_1}{\rightarrow}...$ 
be a run on $\game \restr \straa'$ with a starting state $(X_0,Y_0,C_0,Z_0)$. 
Then for all $n\geq0$ the following assertions hold:
\begin{enumerate}
 \item $X_{n+1}\in \Supp(\trans(X_n,W_n))$.
 \item $Z_n(X_n)$ is not empty.
 \item $Z_{n+1}(X_{n+1})\subseteq Z_n(X_n)$.  
 \item $(Y_n,C_n,Z_n)\stackrel{W_n}{\rightarrow}(Y_{n+1},C_{n+1},Z_{n+1})$ is an edge in $E$, where $(V,E)=\PG(\straa)$.
 \item If $C_n(X_n) =1$, then $C_{n+1}(X_{n+1})=1$. 
 \item If $C_n(X_n)=1$, then $|Z_n(X_n)|=1$. 
  If $\lbrace Z\rbrace = Z_n(X_n)$, then for all $j \geq 0$ we have 
  $\col(X_{n+j})\in Z$.
\end{enumerate}
\end{lemma}
\begin{proof}
We prove all the points below.
\begin{enumerate}
 \item The first point is a direct consequence of the definition of the Markov chain.
 \item The second point follows also from the definition of the chain as 
 from every state of a Markov chain at least one recurrent class is reachable and hence the projection on colors 
 is not empty.
 \item The third point follows from the first point of the lemma and 
 Lemma~\ref{l-Z-decrease}.
 \item For the fourth point: given $(X_n,Y_n,C_n,Z_n)$, the strategy $\straa_n'$ samples $W_n$ with uniform probability among the elements of the set:
\[\lbrace a\in\Act\ |\ \exists m\in M\ s.t.\ (Y_n,C_n,Z_n)\stackrel{a}{\rightarrow}m\in E\rbrace.\]
Once $W_n$ has been chosen, then $\straa_u'((Y_n,C_n,Z_n),\obsmap(X_{n+1}),W_n)$ samples $(Y_{n+1},C_{n+1},Z_{n+1})$ uniformly among the elements of the set:
\[\lbrace (U,B,L)\in \powset(S)\times \set{0,1}^{|S|}\times \D^{|S|} \mid
(Y_n,C_n,Z_n)\stackrel{W_n}{\rightarrow} (U,B,L)\in E \mbox{ and } U \subseteq
\obsmap^{-1}(\obsmap(X_{n+1})) \rbrace.\]
This proves that $(Y_n,C_n,Z_n)\stackrel{W_n}{\rightarrow} (Y_{n+1},C_{n+1},Z_{n+1})$ is an edge in $E$.
 \item The fifth point follows from the fourth point and Lemma \ref{l-C-croit}.
 \item Suppose $(X_n,Y_n,C_n,Z_n)$ is such that $C_n(X_n)=1$. Let $m\in M$ be an arbitrary memory state such that $C_n=\BoolRec_\straa(m)$ and $Z_n=\SetRec_\straa(m)$. 
By hypothesis, since $C_n(X_n)=1$, it follows that $(X_n,m)$ is a recurrent state in the Markov chain $\game \restr \straa$. 
As a consequence, only one recurrent class $R \subseteq S \times M$ of $\game \restr \straa$ is reachable from $(X_n,m)$, 
and $(X_n,m)$ belongs to this class (Property~2 of Markov chains). 
Hence $Z_n(X_n)=\set{\col(\proj_1(R))}$, and thus $|Z_n(X_n)|=1$.
It also follows that all states $(X',m')$ reachable in one step from $(X_n,m)$ also belong to the recurrent class
$R$. It follows that $X_{n+1} \in \proj_1(R)$ and hence $\col(X_{n+1}) \in \col(\proj_1(R))$. By induction for all $j \geq 0$ we have 
$\col(X_{n+j}) \in \col(\proj_1(R))$.
\end{enumerate}
The desired result follows.
\end{proof}

We now introduce the final notion that is required to complete the proof.
The notion is that of a pseudo-recurrent state.
Intuitively a state $(X,Y,C,Z)$ is pseudo-recurrent if $Z$ contains 
exactly one recurrent subset, $X$ belongs to the subset and it will follow
for some memory $m \in M$ (of certain desired property) $(X,m)$ is a 
recurrent state in the Markov chain $\game \restr \straa$.
The important property that will be useful is that once a pseudo-recurrent state
is reached, then $C$ and $Z$ remain invariant.
We now formally define pseudo-recurrent states.

\begin{definition}[Pseudo-recurrent states]
 Let $X\in S$, $Y\subseteq S$, $C\in\set{0,1}^{|S|}$, and  $Z\in\D^{|S|}$.
Then the state $(X,Y,C,Z)$ is called \emph{pseudo-recurrent} if there exists 
$Z_\infty\subseteq \priority$ such that:
\begin{center}
(i)~$Z(X)= \set{ Z_\infty}$, (ii)~$\col(X)\in Z_\infty$, and (iii)~$C(X)=1$.
\end{center}
\end{definition}

The following lemma shows that in the Markov chain $\game \restr \straa'$, all states reachable from a 
pseudo-recurrent state are also pseudo-recurrent.

\begin{lemma}\label{l-reach-pr-pr}
Let $(X,Y,C,Z)$ be a pseudo-recurrent state. If $(X',Y',C',Z')$ is reachable 
from $(X,Y,C,Z)$ in $\game \restr \straa'$, then $(X',Y',C',Z')$ 
is also a pseudo-recurrent state and $Z'(X')=Z(X)$. 
\end{lemma}
\begin{proof}
 Let $(X,Y,C,Z)$ be a pseudo-recurrent state, and let $Z_\infty\subseteq \priority$ be such that $Z(X) = \set{Z_\infty}$, $\col(X)\in Z_\infty$, and $C(X)=1$. 
By Lemma~\ref{l-chain_Z} (fifth point), if $C(X)=1$, then $C'(X')=1$. 
By Lemma~\ref{l-chain_Z} (third point) also, 
$Z'(X')= \set{Z_\infty}$, since $Z'(X')$ is a non empty subset of $Z(X)$. Finally, the fact that $\col(X')\in Z_\infty$ follows from the last (sixth) point of 
Lemma~\ref{l-chain_Z}.
\end{proof}


In the following lemma we show that with probability~1 a pseudo-recurrent state is
reached in $\game \restr \straa'$.

\begin{lemma}\label{l-reach-pr-prob-1}
Let $X\in S$, $Y\in\powset(S)$, $C\in\set{0,1}^{|S|}$, and $Z\in\D^{|S|}$.
Then almost-surely (with probability~1) a run on $\game \restr \straa'$ from 
any starting state $(X,Y,C,Z)$ reaches a pseudo-recurrent state.
\end{lemma}
\begin{proof}
We show that given $(X,Y,C,Z)$ there exists a pseudo-recurrent state 
$(X',Y',C',Z')$ which is reachable from $(X,Y,C,Z)$ in $\game \restr \straa'$.
First let us consider the Markov chain $\game \restr \straa$ obtained from the
original finite-memory strategy $\straa$ with memory $M$. 
Let $m\in M$ be such that $C=\BoolRec_{\straa}(m)$ and $Z=\SetRec_\straa(m)$.
We will now show that the result is a consequence of Lemma~\ref{l-decroit}.
First we know that there exists $t\in S$ and $m'\in M$ such that 
$(t,m')$ is recurrent and reachable from $(X,m)$ with positive probability in
$\game \restr \straa$. Let $R \subseteq S \times M$ be the unique recurrent
class such that $(t,m') \in R$, and 
$Z_\infty = \set{\col(\proj_1(R))}$.
By Lemma~\ref{l-decroit}, this implies that from $(X,Y,C,Z)$ we can reach a 
state $(X',Y',C',Z')$ such that:
\begin{itemize}
 \item $X'=t$;
 \item $Z'(X')=\set{Z_\infty}$;
 \item $\col(X')\in Z_\infty$; and
 \item $C'(X')=1$.
\end{itemize}
Hence $(X',Y',C',Z')$ is a pseudo-recurrent state.
This shows that from all states with positive  probability a pseudo-recurrent
state is reached, and since it holds for all states with positive 
probability, it follows that it holds for all states with probability~1 
(Property~1~(a)).
\end{proof}

In the following three lemmas we establish the required properties of 
pseudo-recurrent states.

\begin{lemma}\label{l-reach-set-rec}
Let $(X,Y,C,Z)$ be a state of $\game \restr \straa'$, and let $Z_B\in Z(X)$. 
Then there exists a pseudo-recurrent state $(X',Y',C',Z')$ which is 
reachable from $(X,Y,C,Z)$ and such that $Z'(X')=\set{Z_B}$.
\end{lemma}
\begin{proof}
The proof is of similar flavor as Lemma~\ref{l-reach-pr-prob-1}.
Consider the Markov chain $\game \restr \straa$ arising by fixing the original
finite-memory strategy $\straa$.
Let $m\in M$ such that $C= \BoolRec_\straa(m)$ and $Z=\SetRec_\sigma(m)$. 
We have $Z_B\in \SetRec_\straa(m)(X)$, hence $Z_B = \col(\proj_{1}(R))$  for some recurrent class $R$
of the chain $\game \restr \straa$ reachable from $(X,m)$. 
Let $t\in S$ and $m'\in M$ be such that $(t,m')$ is reachable from $(X,m)$ in
$\game \restr \straa$ and 
$\Rec(\game \restr \straa)((t,m')) = \set{R}$, then 
$\SetRec_\straa(m')(t)=\set{Z_B}$.
By Lemma~\ref{l-decroit}, there exists $Y',C'$ such that 
$(t,Y',C',\SetRec_\straa(m'))$ is reachable from 
$(X,Y,C,\SetRec_\straa(m))$ in $\game \restr \straa'$ from the 
starting state $(X,Y,C,Z)$.
The desired result follows.
\end{proof}

\begin{lemma}\label{l-prelim-lim}
Let $(X,Y,C,Z)$ be a pseudo-recurrent state, and $Z_\infty\subseteq \priority$ such that $Z(X)=\lbrace Z_\infty\rbrace$. 
Then for all colors 
$\ell\in Z_\infty$, there exists a state 
$(X',Y',C',Z')$ which is reachable in $\game \restr \straa'$ from 
$(X,Y,C,Z)$ and such that $\col(X')=\ell$.
\end{lemma}
\begin{proof}
We again consider the Markov chain $\game \restr \straa$.
Let $m\in M$ be such that $C=\BoolRec_\straa(m)$ and $Z=\SetRec_\straa(m)$. Let $R$ be the unique recurrent class in $\game \restr \straa$ such that $(X,m) \in R$, 
then $Z_\infty = \col(\proj_{1}(R))$. 
For every $\ell\in Z_\infty$, there exists $(X',m') \in R$ such that $\col(X') =\ell$. As $(X',m')$ is reachable from $(X,m)$ in $\game \restr \straa$, by Lemma~\ref{l-decroit},
there exists $Y',C',Z'$ such that $(X',Y',C',Z')$ is reachable 
in $\game \restr \straa'$ from $(X,Y,C,Z)$.
\end{proof}


\begin{lemma}\label{l-rec_pr_pr}
Let $(X,Y,C,Z)$ be a pseudo-recurrent state, then we have 
$Z(X)=\SetRec_{\straa'}(m')(X)$, where $m'=(Y,C,Z)$.
In other words, if we consider a pseudo-recurrent state, and then consider
the projection on the state space of the POMDP $\game$ of the recurrent 
classes reachable and consider the colors, 
then they coincide with $Z(X)$. 
\end{lemma}
\begin{proof}
Let  $(X,Y,C,Z)$ be a pseudo-recurrent state, and let $Z_\infty$ be such that $Z(X)=\lbrace Z_\infty\rbrace$.
First, by Lemma~\ref{l-reach-pr-pr}, we know that if $(X',Y',C',Z')$ is 
reachable from $(X,Y,C,Z)$ in $\game \restr \straa'$, then $\col(X')\in Z_\infty$. 
This implies that for all $Z_B\in \SetRec_{\straa'}(m')(X)$, where $m'=(Y,C,Z)$, 
we have $Z_B\subseteq Z_\infty$.
Second, by Lemma \ref{l-prelim-lim}, if $(X',Y',C',Z')$ is reachable from 
$(X,Y,C,Z)$ in $\game \restr \straa'$ and $\ell\in Z_\infty$, then there 
exists $(X'',Y'',C'',Z'')$ reachable from $(X',Y',C',Z')$  
such that $\col(X'')=\ell$. This implies that for all $Z_B\in \SetRec_{\straa'}(m')(X)$, 
where $m'=(Y,C,Z)$, we have $Z_\infty\subseteq Z_B$.
Thus, $\SetRec_{\straa'}(m')(X)=\lbrace Z_\infty\rbrace=Z(X)$.
\end{proof}




We are now ready to prove the main lemma which shows that the color sets
of the projections of the recurrent classes on the state space of the POMDP coincide for 
$\straa$ and the projected strategy $\straa'=\projp(\straa)$.

\begin{lemma}\label{l-equiv-recs}
Consider a finite-memory strategy $\straa=(\straa_u,\straa_n,M,m_0)$ and
the projected strategy $\straa'=\projp(\straa)=(\straa_u',\straa_n',M',m_0')$.
Then we have
\[
\SetRec_{\straa'}(m_0')(s_0)=\SetRec_\straa(m_0)(s_0);
\]
i.e., the colors 
of the projections of the recurrent classes of the two strategies on the 
state space of the POMDP $\game$ coincide.
\end{lemma}
\begin{proof}
For the proof, let $X=s_0$, $Y=\lbrace s_0\rbrace$, $C=\BoolRec_\straa(m_0)$, 
$Z=\SetRec_\straa(m_0)$. We need to show that 
$\SetRec_{\straa'}(m_0')(X)=Z(X)$, where $m_0'=(Y,C,Z)$.
We show inclusion in both directions.

\begin{itemize}
\item \emph{First inclusion:($Z(X) \subseteq \SetRec_{\straa'}(m_0')(X)$).}
Let $Z_B\in Z(X)$. By Lemma \ref{l-reach-set-rec}, there exists a state
$(X',Y',C',Z')$ which is reachable in $\game \restr \straa'$ from 
$(X,Y,C,Z)$, which is 
pseudo-recurrent, and such that $Z'(X')=\lbrace Z_B\rbrace$. 
By Lemma \ref{l-rec_pr_pr}, we have $Z'(X')=\SetRec_{\straa'}(m')(X')$ where 
$m'=(Y',C',Z')$. By Lemma \ref{l-incl-rec}, we have 
$\SetRec_{\straa'}(m')(X')\subseteq \SetRec_{\straa'}(m_0')(X)$. 
This proves that $Z_B\in \SetRec_{\straa'}(m_0')(X)$.

\item \emph{Second inclusion: ($\SetRec_{\straa'}(m_0')(X) \subseteq Z(X)$).}
Conversely, let $Z_B\in \SetRec_{\sigma'}(m_0')(X)$. 
Since $\game \restr \straa'$ is a finite Markov chain, 
there exists a state $(X',Y',C',Z')$ which is reachable from $(X,Y,C,Z)$ 
in $\game \restr \straa'$ such that:
\begin{itemize}
 \item $\lbrace Z_B\rbrace=\SetRec_{\straa'}(m')(X')$, where $m'=(Y',C',Z')$.
 \item For all $(X'',Y'',C'',Z'')$ reachable from $(X',Y',C',Z')$ 
  in $\game \restr \straa'$ we have 
  $\lbrace Z_B\rbrace=\SetRec_{\straa'}(m'')(X'')$ where $m''=(Y'',C'',Z'')$.
\end{itemize}
The above follows from the following property of a finite Markov chain:
given a state $s$ of a finite Markov chain and a recurrent class $R$ 
reachable from $s$, from all states $t$ of $R$ the recurrent class reachable 
from $t$ is $R$ only (Property~2 of Markov chains). The condition is preserved by a projection on colors 
of states in $R$.

By Lemma \ref{l-reach-pr-prob-1}, there exists a pseudo-recurrent state 
$(X'',Y'',C'',Z'')$ 
which is reachable from $(X',Y',C',Z')$ in $\game \restr \straa'$.
By Lemma \ref{l-rec_pr_pr}, we know that $Z''(X'')=\SetRec_{\sigma'}(m'')(X'')$ 
where $m''=(Y'',C'',Z'')$. 
Since $\SetRec_{\straa'}(m'')(X'')=\lbrace Z_B\rbrace$, and since by 
Lemma~\ref{l-chain_Z} (third point) we have 
$Z''(X'')\subseteq Z'(X')\subseteq Z(X)$, we get that $Z_B\in Z(X)$.
\end{itemize}
The desired result follows.
\end{proof}

\begin{theorem}\label{thrm:muller}
Given a POMDP $G$ and a Muller objective $\Muller(\calf)$ 
with the set $\priority$ of colors,
the following assertions hold:
\begin{enumerate}
\item If there is a finite-memory almost-sure winning 
strategy $\sigma$, then the projected strategy $\projp(\sigma)$,
with memory of size at most $\Mem^*= 2^{2\cdot |S|} \cdot |\D|^{|S|}$  (where $\D
= \powset(\powset(\priority))$), is also an almost-sure winning strategy.

\item If there is a finite-memory positive winning 
strategy $\sigma$, then the projected strategy $\projp(\sigma)$,
with memory of size at most $\Mem^*$, is also a positive 
winning strategy.

\item Finite-memory almost-sure (resp. positive) winning 
strategies require at least exponential memory in general, and 
randomized belief-based stationary strategies are not sufficient in general 
for finite-memory almost-sure and positive winning strategies.
\end{enumerate}
\end{theorem}

\begin{proof}
Consider a finite-memory strategy $\sigma$ with memory $M$ 
and initial memory $m_0$ and the Markov chain $G \restr \sigma$.
Also recall that the number of memory states used by $\projp(\sigma)$ 
is at most $\Mem^*$.
\begin{enumerate}
\item 
By Lemma~\ref{lemm:Markov-basic} if $\sigma$ is almost-sure 
winning, then for all recurrent classes $C$ reachable from $(s_0,m_0)$ 
in $G \restr \sigma$ we have $\col(\proj_1(C)) \in \calf$; and 
by Lemma~\ref{l-equiv-recs} it follows that in the Markov chain 
$G \restr \projp(\sigma)$ for all recurrent classes $C'$ 
reachable from $(s_0,m_0')$, where $m_0'$ is the initial 
memory of $\projp(\sigma)$, we have $\col(\proj_1(C')) \in \calf$. 
It follows from Lemma~\ref{lemm:Markov-basic} that $\projp(\sigma)$
is an almost-sure winning strategy.

\item 
By Lemma~\ref{lemm:Markov-basic} if $\sigma$ is positive
winning, then there exists some recurrent class $C$ reachable from 
$(s_0,m_0)$ in $G \restr \sigma$ with $\col(\proj_1(C))\in \calf$; and 
by Lemma~\ref{l-equiv-recs} it follows that in the Markov chain 
$G \restr \projp(\sigma)$ there exists some recurrent class $C'$ 
reachable from $(s_0,m_0')$, where $m_0'$ is the initial 
memory of $\projp(\sigma)$, with $\col(\proj_1(C')) \in \calf$. 
It follows from Lemma~\ref{lemm:Markov-basic} that $\projp(\sigma)$
is a positive winning strategy.

\item The exponential memory requirement follows from the results
of~\cite{CDH10a} that shows exponential memory requirement for 
almost-sure winning strategies for reachability objectives 
and positive winning strategies for safety objectives.
The fact that randomized belief-based stationary strategies are not sufficient 
follows from Example~\ref{ex:belief-suff}.
\end{enumerate} 
The desired result follows.
\end{proof}


\newcommand{\Int}{I}
\newcommand{\wt}{\widetilde}
\newcommand{\wb}{\overline}

\section{Strategy Complexity for Parity Objectives under 
Finite-memory Strategies}
\label{sec:reduction}
In this section we will establish the exponential upper bounds for almost-sure
(resp. positive) winning strategies in POMDPs with parity objectives 
under finite-memory strategies.
Observe that B\"uchi and coB\"uchi objectives are parity (hence also Muller) 
objectives with $2$ priorities (or colors) (i.e., $d=2$), and from 
Theorem~\ref{thrm:muller} we already obtain an upper bound of 
$2^{6 \cdot |S|}$ on memory size for them.
However, applying the general result of Theorem~\ref{thrm:muller} for Muller
objectives to parity objectives we obtain a double exponential bound, and we
will establish the exponential memory bound.
Formally we will establish Theorem~\ref{thrm:parity}:
for item~(1), in Section~\ref{subsec:positive} we present a reduction that
for positive winning given a POMDP with $|S|$ states and a parity objective 
with $2\cdot d$ priorities constructs an equivalent POMDP with $d \cdot |S|$ 
states with B\"uchi objectives (and thus applying Theorem~\ref{thrm:muller} we 
obtain the $2^{3\cdot d \cdot |S|}$ upper bound); 
for item~(2), in Section~\ref{subsec:almost} we present a reduction that
for almost-sure winning given a POMDP with $|S|$ states and a parity objective 
with $2\cdot d$ priorities constructs an equivalent POMDP with $d \cdot |S|$ 
states with a coB\"uchi objective (and thus applying Theorem~\ref{thrm:muller} we obtain 
the $2^{3\cdot d \cdot |S|}$ upper bound); 
and item~(3) follows as in the proof of Theorem~\ref{thrm:muller}.

\begin{theorem}\label{thrm:parity}
Given a POMDP $G$ and a parity objective $\Parity(p)$ 
with the set $\priority$ of $d$ priorities,
the following assertions hold:
\begin{enumerate}

\item If there is a finite-memory positive winning strategy, then 
there is a positive winning strategy with memory of size at most $2^{3\cdot d \cdot |S|}$.

\item If there is a finite-memory almost-sure winning strategy, then 
there is an almost-sure winning strategy with memory of size at most $2^{3\cdot d \cdot |S|}$.

\item Finite-memory almost-sure (resp. positive) winning strategies require 
exponential memory in general, and belief-based stationary strategies are not sufficient 
in general for finite-memory almost-sure (resp. positive) winning strategies.

\end{enumerate}
\end{theorem}

\subsection{Positive parity to positive B\"uchi}\label{subsec:positive}
Given a POMDP $\game = (S, \Act, \trans, \Obs, \obsmap, s_0)$ and a parity
objective $\Parity(p)$ with priority set $\priority = \{0,\ldots, 2\cdot d \}$, 
we will construct a POMDP $\game' = (S', \Act, \trans', \Obs', \obsmap',s'_0)$ 
together with a B\"uchi objective $\Buchi(\target)$ such that positive winning
under finite-memory strategies is preserved.
Let $\Int$ denote the set $\{0, \ldots, d\}$.
Intuitively, in the construction of $\game'$ we form a copy $\game_i$ 
of the POMDP $\game$ for each $i \in \Int$. 
There will be a positive probability of going from the newly added initial 
state $s'_0$ to every copy $\game_i$.
The transition probabilities in the copy $\game_i$ will mimic the transition 
probability of $\game$ for states with priority at least $2\cdot i$, and for 
states with priority strictly smaller than $2\cdot i$ it mimics the transitions 
of $\game$ with probability $\half$ and with probability $\half$ goes
to a newly added absorbing state $s_f$.
Note that the construction will ensure that for any finite-memory 
strategy, in the copy $\game_i$ there are no recurrent classes that contain 
priorities strictly smaller than $2\cdot i$ as the absorbing state $s_f$
is always reached with positive probability from such states  
(with priority strictly smaller than $2\cdot i$). 
Note that every recurrent class that intersects with an 
absorbing state must consists only of the absorbing state, since
there are no transitions from the absorbing state to any other
state.
In the copy $\game_i$ states with priority $2\cdot i$ are assigned priority~0,
and all other states are assigned priority~1.
Formally the construction is as follows:
\begin{itemize}
\item $S' = (S \times \Int) \cup \set{s'_{0}, s_{f}}$.
\item We define the probabilistic transition function $\trans'$ as follows:
\begin{enumerate}
\item $\trans'(s'_0, a)((s,i)) = \frac{\trans(s_0,a)(s)}{\vert I \vert}$, 
for all $a \in \Act$ and all $i \in \Int$, i.e., with positive probability 
we move to copy $\game_i$ for all $i \in \Int$;
\item $\trans'((s,i),a)((s',i)) = 
\begin{cases}
\trans(s,a)(s') & \mbox{if } p(s) \geq 2\cdot i; \\[2ex]
\frac{\trans(s,a)(s')}{2} & \mbox{otherwise; }\\
\end{cases}$
\item if $p(s) < 2\cdot i$, then we also have $\trans'((s,i),a)(s_f) = \frac{1}{2}$;
\item $\trans'(s_f,a)(s_f) = 1$ for all $a \in \Act$ (i.e., $s_f$ is an absorbing state).
\end{enumerate}
\item $\Obs' = \Obs \cup \{s_f\}$.
\item $\obsmap'((s,i)) = \obsmap(s)$, $\obsmap'(s'_0) = \obsmap(s_0)$ and
$\obsmap'(s_f) = s_f$.
\end{itemize}
We define the priority function $p'$ for the B\"uchi objective as follows:
\begin{enumerate}
\item $p'(s'_0) = 1$;
\item $p'((s,i)) =
\begin{cases}
0 & \mbox{if } p(s) = 2\cdot i; \\
1 & \mbox{otherwise;} \\
\end{cases}$
\item $p'(s_f) = 1$.
\end{enumerate}
The B\"uchi objective for $\game'$ is $\Buchi(p'^{-1}(0))$, i.e., the target 
set $\target$ is the set of states with priority~0 according to $p'$.

\begin{lemma}\label{lemm:pos1}
If there exists a finite-memory positive winning strategy in $\game$ for the
parity objective $\Parity(p)$, then there exists a finite-memory positive 
winning strategy with the same memory states in $\game'$ for the 
objective $\Buchi(p'^{-1}(0))$.
\end{lemma}
\begin{proof}
Let $\straa  = (\straa_u, \straa_n, M, m_0)$ be a finite-memory positive winning 
strategy in the POMDP $\game$ for the objective $\Parity(p)$.
We define the strategy $\straa' = (\straa_u',\straa'_n, M, m_0)$ in the POMDP $\game'$ as follows:
the strategy plays as $\straa$ for all states other than $s_f$, and $\straa'_u(\{s_f\},m,a)(m)$
for all $a \in \Act$.
As the only state in the observation $s_f$ is the absorbing state $s_f$, no
matter what the strategy plays, $s_f$ is not left. The rest of the components is
the same as in the strategy $\straa$.
Let $\wh{\game}$ denote the Markov Chain $\game \restr \straa$ and $\wh{\game}'$ the Markov
chain $\game' \restr \straa'$.

\smallskip\noindent{\em Reachability.}
We first show that if $(s',m')$ is reachable from $(s,m)$ in $\wh{\game}$  
for $s,s'\in S$ and $m,m'\in M$, then for all $i\in \Int$ we have that 
$(s',i,m')$ is reachable from $(s,i,m)$ in $\wh{\game}'$. 
We prove the fact inductively. Let $(s,m) \rightarrow (s',m')$ be an edge in $\wh{G}$,
then there exists an action $a \in \Act$ such that 
(i)~$s' \in \Supp(\trans(s,a))$, (ii)~$a \in \Supp(\straa_n(m))$, 
and (iii)~$m' \in \Supp(\straa_u(m,\obsmap(s'),a))$.
By definition of the transition function $\trans'$ this implies that 
$(s',i) \in \Supp(\trans'((s,i),a)$.
Since $\straa'$ plays the same as $\straa$,  
it follows that $(s,i,m) \rightarrow (s',i,m')$ is an edge in $\wh{\game}'$. 
As the state spaces of the Markov chains are finite, this extends to 
reachability by simple induction.

\smallskip\noindent{\em Recurrent class.}
Since $\straa$ is a positive winning strategy, there must exist a recurrent
class $C$ reachable from $(s_0,m_0)$ in $\wh{\game}$ such that the  minimal 
priority $\min(p(\proj_1(C)))$ is even. 
Let that minimal priority be $2 \cdot i$.
Consider the copy $\game_i$ of $\game$ in $\game'$: it contains all states from $\proj_1(C)$,
and moreover as the minimal priority of the states in $\proj_1(C)$ is 
$2 \cdot i$ (according to $p$), 
the transition function $\trans'$ matches the transition function $\trans$ 
for states in $\proj_1(C)$ and all actions $a \in \Act$.
As the strategy $\straa'$ does not know, due to the observation definition, in
which copy it is and plays as the strategy $\straa$ does in
$\game$, the set $C' = \{(s,i,m) \mid (s,m) \in C \}$ of
states forms a recurrent class in $\wh{\game}'$.

Finally we need to show that $C'$ is reachable from $(s_0',m_0)$ in $\wh{\game}'$.
Since $C$ is reachable from $(s_0,m_0)$ in $\wh{\game}$, there exists a state 
$(s,m)$ that is reachable in one step from $(s_0,m_0)$ and $C$ is reachable from 
$(s,m)$.
The state $(s,i,m)$ is reachable in one step from $(s_0',m_0)$ in $\wh{\game}'$ 
(from the initial state $s_0'$ all copies are reached with positive probability), 
and reachability to $C'$ from $(s,i,m)$ follows from the argument above for 
reachability.
As the set $\proj_1(C')$ contains a state $s$ with $p(s)= 2 \cdot i$, we
have that $p'((s,i)) = 0$, i.e., a target state belongs to $C'$. 
It follows that $\straa'$ is a finite-memory positive winning strategy
in $\game'$ for $\Buchi(p'^{-1}(0))$ and the desired result follows.
\end{proof}

\begin{lemma}\label{lemm:pos2}
If there exists a finite-memory positive winning strategy in $\game'$ for the 
objective $\Buchi(p'^{-1}(0))$, then there exists a finite-memory positive winning
strategy with the same memory states in $\game$ for the objective $\Parity(p)$.
\end{lemma}
\begin{proof}
Given a finite-memory positive winning strategy 
$\straa=(\straa_u,\straa_n,M,m_0)$ in the POMDP $\game'$
we show that $\straa$ is also positive winning in the POMDP $\game$.
Similar to the previous lemma we fix the strategy $\straa$ in $\game$ and
obtain a Markov Chain $\wh{\game} = \game \restr \straa$ and 
$\wh{\game}' = \game'\restr \straa$. 

\smallskip\noindent{\em Reachability.}
We show that if  $(s',i,m')$ is reachable from $(s,i,m)$ in $\wh{\game}'$ 
for $s,s' \in S$, $m,m'\in M$, and $i \in \Int$, then $(s',m')$ is reachable 
from $(s,m)$ in $\wh{\game}$. 
This follows from the fact that (i)~if there is an edge 
$(s,i) \stackrel{a}{\rightarrow}(s',i)$ in $\wh{\game}'$ 
(i.e., $\trans'((s,i),a)((s',i))>0$), then we have an edge 
$s\stackrel{a}{\rightarrow}s'$ in $\wh{\game}$ (i.e., $\trans(s,a)(s')>0$);
and (ii)~the strategy $\straa$ is the same for both POMDPs.

\smallskip\noindent{\em Recurrent class.}
As $\straa$ is a positive winning strategy in $\game'$, there exists a recurrent
class $C'$ reachable from $(s_0',m_0)$ in $\wh{\game}'$ which satisfies 
$\proj_1(C') \cap p'^{-1}(0)$ is non-empty. 
Note that there must exist an $i \in \Int$ such that all the states of the 
recurrent class $C'$ are elements from the set $S \times \{i\} \times M$, 
i.e., the class is included in some copy $\game_i$ (since there are no 
transitions between copies and the absorbing state has priority~$1$).
As $C'$ forms a recurrent class in copy $i$ it follows that all the states in $\proj_1(C')$ 
have priority at least $2 \cdot i$ according to $p$ 
(since states with priority strictly smaller than $2\cdot i$ according to 
$p$ have positive transition probability to $s_f$).
Consider the set of states $C = \{(s,m) \mid
(s,i,m) \in C' \}$ in $\wh{\game}$. As the strategies we consider are the same,
$C$ forms a recurrent class in $\wh{\game}$ with the minimal priority
at least $2 \cdot i$. 
Moreover, since $\proj_1(C') \cap p'^{-1}(0)$ is non-empty, it follows that the
minimal priority of $C$ is exactly $2\cdot i$, i.e., 
$\min(p(\proj_1(C)))$ is $2\cdot i$ and even.

Finally, it remains to show that $C$ is reachable from the initial state of $\wh{\game}$.
Since $C'$ is reachable from $(s_0',m_0)$, it must be reachable from some state 
$(s,i,m)$ of copy $\game_i$ and $(s,i,m)$ is reachable in one step from $(s_0',m_0)$
in $\wh{\game}'$.
Then it follows that the state $(s,m)$ is reachable from $(s_0,m_0)$ in  
one step in $\wh{\game}$, and the reachability of $C$ from $(s,m)$ follows 
from the fact that $C'$ is reachable from $(s,i,m)$ and the argument for 
reachability above.
Hence, $\straa$ is a positive winning strategy for the objective $\Parity(p)$
in $\game$ and the desired result follows.
\end{proof}

Lemma~\ref{lemm:pos1} and Lemma~\ref{lemm:pos2} establishes item~(1)
of Theorem~\ref{thrm:parity}.

\subsection{Almost-sure parity to almost-sure coB\"uchi}\label{subsec:almost}
For almost-sure winning the reduction from parity objectives to coB\"uchi objectives
will be achieved in two steps: (1)~First we show how to reduce POMDPs with parity objectives
to POMDPs with parity objectives with priorities in $\set{0,1,2}$; and (2)~then show how to 
reduce POMDPs with  priorities in $\set{0,1,2}$ to coB\"uchi objectives, for
almost-sure winning.

\subsubsection{Almost-sure parity to almost-sure parity with three priorities}
Given a POMDP $\game = (S, \Act, \trans, \Obs, \obsmap, s_0)$ and a parity
objective $\Parity(p)$ with priority set $\priority = \{0,\ldots, 2 \cdot d+1 \}$, 
we will construct a POMDP $\wb{\game} = (\wb{S}, \Act, \wb{\trans}, \Obs,
\wb{\obsmap},
\wb{s}_0)$ together with a parity objective $\Parity(\wb{p})$ which assigns priorities
from the set $\{0,1,2\}$ such that almost-sure winning under finite-memory 
strategies is preserved.
Let $\Int$ denote the set $\{0, \ldots, d\}$.
Intuitively to construct $\wb{\game}$ we form a copy $\wb{\game}_i$ of the POMDP $\game$ for each
$i \in \Int$. The game starts in the initial state of the copy $\wb{\game}_d$. 
The transition probabilities in the copy $\wb{\game}_i$ are as follows:
for states with priority at least $2\cdot i$ it mimics the transition of 
$\game$; and for states with priority strictly smaller than $2\cdot i$ it 
mimics the transition of $\game$ with probability $\half$ and with probability 
$\half$ moves to the copy $i-1$ (i.e., to $\wb{\game}_{i-1}$).
In $\wb{\game}_{i}$, states with priority $2\cdot i$ and $2\cdot i+1$ are assigned 
priorities~0 and~1, respectively, and all other states are assigned priority~2.
We now present the formal construction of $\wb{\game}$:

\begin{itemize}
\item $\wb{S} = S \times \Int$
\item We define the transition function $\wb{\trans}$ in two steps; for a state
$(s,i) \in S \times \Int$ and an action $a \in \Act$:
\begin{enumerate}
\item $\wb{\trans}((s,i),a)((s',i)) = 
\begin{cases}
\trans(s,a)(s') & \mbox{if } p(s) \geq 2\cdot i \\ 
\frac{\trans(s,a)(s')}{2} & \mbox{otherwise};
\end{cases}$
\item $\wb{\trans}((s,i),a)((s',i-1)) = \frac{\trans(s,a)(s')}{2}$ if 
$p(s) < 2\cdot i$ 
\end{enumerate}
\item $\wb{\obsmap}((s,i)) = \obsmap(s)$;
\item $\wb{s}_0 = (s_0,d)$.
\end{itemize}
The new parity objective $\Parity(\wb{p})$ assigning priorities $\{0,1,2\}$ is
defined as follows:
$$ \wb{p}((s,i)) =
\begin{cases}
0 & \mbox{if } p(s) = 2\cdot i; \\
1 & \mbox{if } p(s) = 2\cdot i + 1; \\
2 & \mbox{otherwise;} \\ 
\end{cases} $$

\begin{lemma}\label{lemm:almost1}
If there exists a finite-memory almost-sure winning strategy in the POMDP
$\game$ for the objective $\Parity(p)$, then there exists a finite-memory 
almost-sure winning strategy with the same memory states in the POMDP 
$\wb{\game}$ for the objective $\Parity(\wb{p})$ with three priorities.
\end{lemma}
\begin{proof}
Let $\straa=(\straa_u,\straa_n,M,m_0)$ be a finite-memory almost-sure winning 
strategy in the POMDP $\game$ for the objective $\Parity(p)$ and $\wh{\game}$ 
the Markov Chain $\game\restr \straa$. 
We show that the strategy $\straa$ is also almost-sure winning in the POMDP 
$\wb{\game}$. 
Consider the Markov Chain $\wh{\game}'  = \wb{\game} \restr\straa$.
We need to show that for every recurrent class $\wb{C}$ reachable from the 
starting state $(s_0,d,m_0)$ we have that $\min(\wb{p}(\proj_1(\wb{C})))$ is even to show that 
$\straa$ is an almost-sure winning strategy in $\wb{\game}$.
We will show that if there is a reachable recurrent class in $\wh{\game}'$
with minimum priority odd, then there is a reachable recurrent class
in $\wh{\game}$ with minimum priority odd contradicting that $\straa$ is 
an almost-sure winning strategy in $\game$.

Assume towards contradiction that there exists a recurrent class 
$\wb{C}$ reachable from $(s_0,d,m_0)$ in $\wh{\game}'$ such that the minimal priority is odd, i.e., 
$\min(\wb{p}(\proj_1(\wb{C})))$ is odd (i.e.,~$\wb{C}$ contains a priority~1 state but no
priority~0 state).
By the construction of $\wb{\game}$, for every copy $\wb{\game}_i$, there are transitions
only to the states in the copy $\wb{\game}_i$ or to the lower copy $\wb{\game}_{i-1}$.
Hence there are no transitions from a lower copy to a higher copy.
Hence every recurrent class in $\wh{\game}'$ must be contained in some copy.
Let the recurrent class $\wb{C}$ be contained in copy $i$, i.e.,
$\wb{C} \subseteq S\times \set{i} \times M$.
Also note that by construction, every state with priority strictly smaller
than $2\cdot i$ (by priority function $p$) has positive probability transition 
to a lower copy and hence such states do not belong to the recurrent class.
Since $\min(\wb{p}(\proj_1(\wb{C})))$ is odd it follows that 
$\wb{C}$ does not contain a state with priority $0$ by $\wb{p}$ (i.e., priority $2\cdot i$ by 
$p$) but contains some state with priority $1$ by $\wb{p}$ (i.e., priority $2\cdot i+1$ by 
$p$), i.e., 
(i)~$\wb{C} \subseteq \big((\bigcup_{j \geq 2\cdot i} p^{-1}(j)) \times \set{i} \times M\big)$ 
($\wb{C}$ is contained in the copy $\wb{\game}_i$ and the priorities of the 
states in $\wb{C}$ are at least $2 \cdot i$);
(ii)~$\wb{C} \cap (p^{-1}(2\cdot i) \times \set{i} \times M)=\emptyset$ ($\wb{C}$ contains no 
priority~0 state according to $\wb{p}$); and
(iii)~$\wb{C} \cap (p^{-1}(2\cdot i+1) \times \set{i} \times M) \neq \emptyset$ 
($\wb{C}$ contains some priority~1 state according to $\wb{p}$).
Observe that due to the definition of observations whenever a state
$(s,i,m)$ is reachable in $\wh{\game}'$ we have that the state $(s,m)$ is also reachable in
$\wh{\game}$ (since memories of the strategies are the same and the observation function
cannot distinguish between copies). 
It follows that the set of states $C = \{(s,m) \mid (s,i,m) \in \wb{C}\}$ are reachable from 
$(s_0,m_0)$ in $\wh{\game}$.
Moreover as transition probabilities for states $(s,j)$ with $j \geq 2\cdot i$ are not modified
in the copy $\wb{\game}_i$ it follows that $C$ is a recurrent class reachable in  $\wh{\game}$.
Thus we have a recurrent class $C$ reachable from $(s_0,m_0)$ in $\wh{\game}$ such that 
$C \cap (p^{-1}(2\cdot i) \times M)=\emptyset$ and
$C \cap (p^{-1}(2\cdot i+1) \times M) \neq \emptyset$, i.e., the minimum
priority is $2 \cdot i+ 1$ and odd.
This contradicts that $\straa$ is an almost-sure winning strategy in $\game$ 
for $\Parity(p)$.
Hence it follows $\straa$ is an almost-sure winning strategy in $\wb{\game}$ for 
$\Parity(\wb{p})$ and the desired result follows.
\end{proof}

\begin{lemma}\label{lemm:almost2}
If there exists a finite-memory almost-sure winning strategy in the POMDP
$\wb{\game}$ for the objective $\Parity(\wb{p})$ with three priorities, then there 
exists a  finite-memory almost-sure winning strategy with the same memory 
states in the POMDP $\game$ for the objective $\Parity(p)$.
\end{lemma}
\begin{proof}
As in the previous lemma let $\straa=(\straa_u,\straa_n,M,m_0)$ be a finite-memory 
almost-sure winning strategy in the POMDP $\wb{\game}$ for the objective
$\Parity(\wb{p})$ 
and $\wh{\game}'$ the Markov Chain $\wb{\game}\restr \straa$. 
We show that the strategy $\straa$ is also almost-sure winning in the POMDP 
$\game$. 
We consider the Markov Chain $\wh{\game}  = \game \restr \straa$.
We need to show that for all recurrent classes $C$ reachable in 
$\wh{\game}$ from $(s_0,m_0)$ the minimal priority is even.

Assume towards contradiction that there exists a reachable recurrent class $C$
from $(s_0,m_0)$ in $\wh{\game}$ with minimal priority odd, and let the minimal 
priority be $2\cdot i+1$. 
We need to show that this implies that there exists a reachable recurrent class 
from $(s_0,d,m_0)$ in $\wh{\game}'$ with minimal priority odd (as we consider 
only priorities $0,1,2$, the minimal priority is $1$).
Consider the subset of states $\wb{C} = \{(s,i,m) \mid (s,m) \in C \}$. The minimal
priority of the set in $\wb{\game}$ is 1 since $C$ does not contain any state with 
priority strictly smaller $2\cdot i+1$ and has some state with priority $2\cdot i +1$. 
The transition function $\wb{\trans}$ matches the transition function $\trans$ on states of 
$\proj_1(C)$ for any action $a\in \Act$. Therefore, $\wb{C}$ forms a recurrent class in $\wh{\game}'$. 
It remains to show that $\wb{C}$ is reachable from the initial state of $\wh{\game}'$.
Let $(s,i,m)$ be a state in $\wb{C}$ such that $\wb{p}(s)=1$: the state $(s,m)$ is reachable in 
$\wh{\game}$ from $(s_0,m_0)$ since $(s,m)$ is a state in the recurrent class $C$ reachable from 
$(s_0,m_0)$ in $\wh{\game}$.
Then for the starting copy $\wb{\game}_d$ we have that $(s,d,m)$ is 
reachable from $(s_0,d,m_0)$ in $\wh{\game}'$.
We now present two simple facts:
\begin{enumerate}
\item 
For all states $(s',m')\in C$ we have that $(s,m)$ is reachable from $(s',m')$ in 
$\wh{\game}$ (since $C$ is a recurrent class and both $(s',m')$ and $(s,m)$ 
belong to $C$), and it follows that for all $j\in \Int$ 
we have that $(s,j,m)$ is reachable from $(s',j,m')$ in the copy $j$.

\item 
Since $\wb{p}((s,i,m))=1$ we have that $p(s)=2\cdot i+1$, and for all $j >i$, 
in $\wb{\game}_j$ if the state $(s,j,m)$ is reached, then with positive 
probability we reach the copy $j-1$ (some state $(s',j-1,m')$).
Moreover, since $(s,m) \in C$, for all $j >i$, from $(s,j,m)$ we reach a 
state $(s',j-1,m')$ such that $(s',m') \in C$. 
\end{enumerate}
From the above two facts it follows that for all $j >i$ we have 
$(s,j-1,m)$ is reachable from $(s,j,m)$. 
It follows that $(s,i,m)$ is reachable from $(s,d,m)$ and since $(s,d,m)$ is
reachable from $(s_0,d,m_0)$ it follows that $(s,i,m)$ is reachable from
$(s_0,d,m_0)$.
Hence $\wb{C}$ is reachable from $(s_0,d,m_0)$ and we have 
a contradiction to the fact that $\straa$ is an almost-sure winning strategy in 
$\wb{\game}$ for $\Parity(\wb{p})$.
It follows that $\straa$ is an almost-sure winning strategy in $\game$ for 
$\Parity(p)$ and the desired result follows.
\end{proof}

Lemma~\ref{lemm:almost1} and Lemma~\ref{lemm:almost2} gives us the following 
lemma.
\begin{lemma}\label{lemm:almost4}
If there exists a finite-memory almost-sure winning strategy $\straa$ in the POMDP
$\game$ with the objective $\Parity(p)$, then $\straa$ is 
an almost-sure winning strategy in the POMDP $\wb{\game}$ with the objective
$\Parity(\wb{p})$ with three priorities; and vice versa.
\end{lemma}

Next we show how to reduce the problem of almost-sure winning for parity
objectives with priorities
from the set $\set{0,1,2}$ to the problem of almost-sure winning for coB\"uchi
objectives in POMDPs.

\subsubsection{Almost-sure parity with three priorities to almost-sure coB\"uchi}
Consider a POMDP $\wb{\game} = (\wb{S}, \Act, \wb{\trans}, \Obs, \wb{\obsmap},
\wb{s}_0)$ with a parity
objective $\Parity(\wb{p})$ assigning priorities from the set $\{0,1,2\}$. We
construct a POMDP $\wt{\game} =  (\wt{S}, \Act, \wt{\trans}, \wt{\Obs},
\wt{\obsmap}, \wt{s}_0)$ with a coB\"uchi objective $\coBuchi(\wt{T})$, where
the set of states $\wt{T}$ is going to be defined as $\wt{p}^{-1}(2)$ for a
function $\wt{p}$ assigning priorities from the set $\{1,2\}$.
Intuitively, for states with priority $1$ and $2$, the transition function 
$\wt{\trans}$ mimics the transitions of $\wb{\trans}$; and for states
with priority~0, the transition function $\wt{\trans}$ mimics the transitions 
of $\wb{\trans}$ with probability~$\half$ and with probability $\half$ goes to 
a newly added absorbing state that is assigned priority~2.
Formally the POMDP $\wt{\game}$ is defined as follows:
\begin{itemize}
\item $\wt{S}= \wb{S} \cup \{\wt{s}_r\}$;
\item $\wt{\trans}$ is defined for all states $s\in \wb{S}$ and all actions $a \in
\Act$ as follows:
\begin{enumerate}
\item $\wt{\trans}(s,a)(s') = 
\begin{cases}
\wb{\trans}(s,a)(s') & \mbox{if } \wb{p}(s) \in \{1,2\}; \\
\frac{\wb{\trans}(s,a)(s')}{2} & \mbox{if } \wb{p}(s) = 0 ; \\
\end{cases} $
\item $\wb{\trans}(s,a)(\wt{s}_r) = \half$ if $\wb{p}(s) = 0$ ;
\item $\wb{\trans}(\wt{s}_r,a)(\wt{s}_r) = 1$, i.e.,
$\wt{s}_r$ is an absorbing state;
\end{enumerate}
\item $\wt{\Obs} = \Obs \cup \{\wt{s}_r\}$, i.e., the additional state is a new
observation;
\item $\wt{\obsmap}(s) = \begin{cases}
\wb{\obsmap}(s) & \mbox{if } s \in \wb{S} \\
\wt{s}_r & \mbox{if } s = \wt{s}_r \\
\end{cases} $ 
\end{itemize}
The coB\"uchi objective is defined by a priority function $\wt{p}$, that is
defined as:
$$\wt{p}(s) = \begin{cases}
\wb{p}(s) & \mbox{if } \wb{p}(s) \in \{1,2\} \\
2 & \mbox{if } \wb{p}(s) = 0  \mbox{ or } s = \wt{s}_r\\
\end{cases} $$

The objective in $\wt{\game}$ is $\coBuchi(\wt{p}^{-1}(2))$. 

\begin{lemma}\label{lemm:almost3}
If there exists a finite-memory almost-sure winning strategy $\straa$ in the POMDP
$\wb{\game}$ with the objective $\Parity(\wb{p})$ with three priorities, then $\straa$ is 
an almost-sure winning strategy in the POMDP $\wt{\game}$ with the objective
$\coBuchi(\wt{p}^{-1}(2))$ and vice versa.
\end{lemma}
\begin{proof}
We start with the first direction. Let $\straa = (\straa_u,\straa_n,M,m_0)$ be a
finite-memory almost-sure winning
strategy in $\wb{\game}$, we claim that
$\straa$ is also almost-sure winning in $\wt{\game}$.
Assume towards contradiction that there exists a reachable recurrent class $C$ in the Markov chain
$\wt{\game} \restr \straa$ such that the minimal priority in the class is $1$.
Then $C$ cannot contain the newly added absorbing state $\wt{s}_r$, as
$\wt{p}(\wt{s}_r) = 2$ and if a recurrent class contains the absorbing state
$\wt{s}_r$, then it contains only the state $\wt{s}_r$ as there is no edge going
from $\wt{s}_r$ to a different state in the POMDP $\wt{\game}$.
 It follows that the set $C$ is reachable in $\wb{\game}
\restr \straa$, and due to the definition of the transition functions forms a
recurrent class. Since $C$ contains a state $s$ with priority $\wt{p}(s) = 1$, we have
that $\wb{p}(s)$ is also  $1$, so there is a state with priority $1$ in $C$. It remains to rule out the possibility that $C$ contains states
$s'$ with priority $\wb{p}(s') = 0$, but that follows from the fact that whenever
there was a state with priority $0$, no matter what action was played, there was
a positive probability of reaching $\wt{s}_{r}$. So $C$ contains a state with
priority $1$ and all the other states have priority $1$ or $2$. It follows
that there exists a reachable recurrent class in $\wb{\game} \restr \straa$, where
the minimal priority is odd. This contradicts our assumption that $\straa$ is
almost-sure winning in $\wb{\game}$.

In the second direction assume that no finite-memory strategy is almost-sure
winning in $\wb{\game}$. Therefore, for every finite-memory strategy $\straa$ there exists a reachable recurrent class $C$ in the Markov
Chain $\wb{\game} \restr \straa$, such that the minimal priority in the class is
$1$, i.e., there exists a state with priority $1$ and there are no states with
priority $0$ in $C$. In the Markov Chain $\wt{\game} \restr \straa$ the transition
functions $\wt{\trans}$ allows every transition available in $\wb{\trans}$. It
follows that the set $C$ is reachable with positive probability in $\wt{\game}
\restr \straa$.
Since there is no state with priority $0$ in $C$ it follows that the transition
function $\wt{\trans}$ matches the transition function $\wb{\trans}$ for all
states in $C$ and all actions $a \in \Act$. It follows that $C$ is a recurrent
class in the POMDP $\wt{\game} \restr \straa$.
 As all the
priorities of the states in $C$ are preserved in the priority function $\wt{p}$, there
exists a reachable recurrent class with minimal priority $1$. It follows that
there is no finite-memory almost-sure winning strategy in the POMDP $\wt{\game}$.
\end{proof}

Lemma~\ref{lemm:almost4} and Lemma~\ref{lemm:almost3} 
establish item~(2) of Theorem~\ref{thrm:parity}.

\newcommand{\pri}{p}

\newcommand{\eqstate}{{\widehat S}}
\newcommand{\eqact}{{\widehat \Act}}
\newcommand{\eqtrans}{{\widehat \trans}}
\newcommand{\eqinitial}{{\widehat s_0}}
\newcommand{\eqobsact}{{\widehat{\Gamma}}}

\newcommand{\eqObs}{{\mathcal {\widehat O}}}
\newcommand{\eqobsmap}{{\widehat \gamma}}
\newcommand{\eqobj}{{\widehat \phi}}
\newcommand{\stact}{\mathcal{C}}
\newcommand{\eqmdp}{{\widehat \game}}
\newcommand{\eqbad}{{\widehat s_b}}
\newcommand{\beliefset}{{\mathcal{B}}}
\newcommand{\obsact}{{\Gamma}}
\newcommand{\obsset}{{O}}
\newcommand{\we}{{\textsf{WE}}}
\newcommand{\allow}{{\textsf{Allow}}}
\newcommand{\pre}{{\textsf{Pre}}}
\newcommand{\apre}{{\textsf{Apre}}}
\newcommand{\obscover}{{\textsf{ObsCover}}}
\newcommand{\Almost}{\mathsf{Almost}}
\newcommand{\Positive}{\mathsf{Positive}}
\newcommand{\rank}{\mathsf{Rank}}
\newcommand{\was}{W_\mathsf{AS}}
\newcommand{\wpos}{W_\mathsf{Pos}}
\newcommand{\buchimem}{M_{\Buchi}}
\newcommand{\cobuchimem}{M_{\coBuchi}}
\newcommand{\almostcobuchi}{{\mathsf{AlmostCoBuchiRed}}}
\newcommand{\wpr}{\mathit{wpr}}

\section{Computational Complexity for Parity Objectives}
In this section we  will present an exponential time algorithm to solve almost-sure 
winning in POMDPs with coB\"uchi objectives under finite-memory strategies
(and the polynomial time reduction of Section~\ref{sec:reduction} for parity objectives
to coB\"uchi objectives allows our results to carry over to parity objectives).
The results for positive B\"uchi is similar and the almost similar proof is
omitted.
The naive algorithm would be to enumerate over all finite-memory strategies 
with memory bounded by $2^{6 \cdot |S|}$, this leads to an algorithm
that runs in double-exponential time. Instead our algorithm consists of two steps:
(1)~given a POMDP $G$ we first construct a special kind of a POMDP $\wh{G}$ 
such that there is a finite-memory winning strategy in $G$ iff there is a 
randomized memoryless winning strategy in $\wh{G}$; and
(2)~then show how to solve the special kind of POMDPs in time polynomial in
the size of the POMDP.
We first introduce the special kind of POMDPs which we call 
belief-observation POMDPs which intuitively satisfy that the current 
belief is always the set of states with current observation.

\begin{definition}[Belief-observation POMDP]
A POMDP $\game=(S,\Act,\trans,\Obs,\obsmap,s_0)$ is a 
\emph{belief-observation POMDP} iff for every finite prefix 
$w=(s_0,a_0,s_1,a_1,\ldots,s_n)$ with the observation sequence 
$\rho=\obsmap(w)$, the belief $\B(\rho)$ is equal to the set of states with 
the observation $\obsmap(s_n)$, 
i.e., $\B(\rho)=\set{s \in S \mid \obsmap(s)=\obsmap(s_n)}$.
In other words, belief-observation POMDPs are the special class of POMDPs
where the current belief can be directly obtained from the current observation.
\end{definition}

\subsection{Construction of belief-observation POMDPs for finite-memory strategies}

\smallskip\noindent{\bf POMDPs to belief-observation POMDPs.}
The goal of this section is  given a POMDP $G$ with a coB\"uchi objective
$\coBuchi(p^{-1}(2))$, and a priority function with priority set $\set{1,2}$, 
to construct a belief-observation POMDP $\eqmdp$ such that if there exists 
a finite-memory almost-sure winning strategy in $\game$,
then there exists a randomized memoryless almost-sure winning strategy in 
$\eqmdp$ for another coB\"uchi objective $\coBuchi(\wh{p}^{-1}(2))$ 
and vice-versa. 
Since we are interested in coB\"uchi objectives, 
for the sequel of this section we will denote by 
$M = 2^{S} \times \{0,1\}^{|S|} \times \D^{|S|}$, i.e., 
all the possible beliefs $\B$, $\BoolRec$ and $\SetRec$ functions 
(recall that $\D$ is  $\powset(\powset(\{1,2\}))$ for coB\"uchi objectives). 
If there exists a finite-memory almost-sure winning strategy $\straa$, then 
the projected strategy $\straa' = \projp(\straa)$ is also a finite-memory 
almost-sure winning strategy (by Theorem~\ref{thrm:muller}) and will use memory 
$M' \subseteq M$.
The size of the constructed POMDP $\eqmdp$ will be exponential in the size of 
the original POMDP $\game$ and polynomial in the size of the memory set $M$ 
(and $|M|=2^{6\cdot |S|}$ is exponential in the size of the POMDP $\game$).
We define the set $\cobuchimem \subseteq M$ as the memory elements, where for 
all states $s$ in the belief component of the memory, the set $\SetRec(s)$ 
contains only a set with priority two, i.e., there is no state with priority~$1$
in the reachable recurrent classes according to $\SetRec$.  
Formally,
\[
\cobuchimem  =  \{ (Y,B,L) \in M \mid \mbox{ for all } s \in Y  \mbox{ we have }
L(s) = \set{\{2\}} \}
\]

\smallskip\noindent{\bf Construction of the new POMDP.}
Given a POMDP $\game=(S,\Act,\trans,\Obs,\obsmap,s_0)$ with  a coB\"uchi
objective $\coBuchi(p^{-1}(2))$, represented by priority function 
$p: S \rightarrow \{1,2\}$,
we construct a new 
POMDP $\eqmdp = (\eqstate,\eqact,\eqtrans, \eqObs, \eqobsmap, \eqinitial)$ with
a coB\"uchi objective $\coBuchi(\wh{p}^{-1}(2))$, for
some priority function $\wh{p}$ assigning to states in $\eqstate$ priorities
from the set $\{1,2\}$. We refer to the newly constructed POMDP $\eqmdp$ as
$\almostcobuchi(\game)$.


\begin{itemize}
\item The set of states $\eqstate =\eqstate_a \cup \eqstate_m \cup \{\eqinitial,
\eqbad\}$ ,
will consist of
\emph{action-selection} states $\eqstate_a \subseteq S  \times M$; 
\emph{memory-selection} states $\eqstate_m \subseteq S  \times 2^S \times M 
\times \Act$;
$\eqinitial$ is an additional initial state; and the state $\eqbad$ is a new absorbing state.

\item The observation set is as follows: 
$\eqObs = (M) \cup (2^{S} \times M \times \Act) 
\cup \{\eqinitial\} \cup \{\eqbad\}$.
\item The initial state of the POMDP is $\eqinitial$.
\item The observation mapping is defined naturally $\eqobsmap((s, m)) =  m$,
$\eqobsmap ((s,Y,m,a)) = (Y,m,a)$, $\eqobsmap(\eqinitial) = \{\eqinitial\}$, and
$\eqobsmap(\eqbad) = \{\eqbad\}$.
In other words, except the states $\eqinitial$ and $\eqbad$ the strategy cannot observe the first 
component of the state.

\item The actions are $\eqact = \Act \cup M$, i.e., the actions from the POMDP $\game$ or memory elements from the set $M$.

\item We define the transition function $\eqtrans$ in the following steps.
First we will introduce a notion of \emph{allowed actions}, observe that for the computation of almost-sure winning under finite-memory 
strategies the precise transition probabilities do not
matter and therefore in the following step we will specify only the edges of the POMDP graph, and all 
transition probabilities are uniform over the support set.

We call an action $a \in \Act$ \emph{allowed} in observation $(Y,B,L) \in \eqObs$ if
for all states $\wh{s} \in Y$, there exists a set $Z_\infty \subseteq \{1,2\}$
such that if $B(\wh{s}) = 1$,
$L(\wh{s})= \set{Z_\infty}$, and  $p(\wh{s}) \in Z_\infty$, then for all states $\wh{s}' \in
\Supp(\trans(\wh{s},a))$ we have $p(\wh{s}') \in Z_\infty$.
Intuitively this condition enforces that once a state that corresponds to a
pseudo-recurrent state is reached in the POMDP $\eqmdp$
in the next step only states with priority in the set $Z_\infty$ can be visited.
Similarly we call an action $(Y',B',L') \in M$ \emph{allowed} in observation $(Y',(Y,B,L),a)$ if both
of the following conditions are satisfied:
(i)~for all states $\wh{s} \in Y$, if $B(\wh{s}) = 1$, then 
for all states  $\wh{s}' \in \Supp(\trans(\wh{s},a))$ we have that $B'(\wh{s}') =1$, 
intuitively the condition says that if the $\BoolRec$ function is set to~1,
then for all successors the $\BoolRec$ function remains~1 (recall 
by Lemma~\ref{l-chain_Z} fifth point the property is ensured for projected 
strategies); and 
(ii)~if $\wh{s} \in S$ and $\wh{s}' \in \Supp(\trans(\wh{s},a))$, we have that
$L'(\wh{s}') \subseteq L(\wh{s})$. Intuitively the condition says the function $\SetRec$ must not increase with
respect to set inclusion along the successors 
(recall by Lemma~\ref{l-chain_Z} third point the property is ensured for 
projected strategies). 


\begin{enumerate}
\item $\eqinitial \stackrel{m}{\rightarrow} (s_0,m)$ for all  $m \in
\cobuchimem \cap \{(\{s_0\},B,L) \mid B \in \{0,1\}^S, L \in \D^S\}$, 
i.e., from the initial state all memory elements from $\cobuchimem$ that 
are consistent with the starting state can be chosen (in other words, it 
consists of all the ways a projected strategy of a finite-memory almost-sure winning strategy
could start);

\item $(s,(Y,B,L)) \stackrel{a}{\rightarrow} (s',Y',(Y,B,L),a)$ iff all of the
following conditions are satisfied:
\begin{itemize}
\item $s' \in \Supp(\trans(s,a))$; and
\item $Y'$ is the belief update in POMDP $\game$ from belief $Y$ under observation $o = \obsmap(s')$ and action $a$,
i.e., $Y' =  \bigcup_{\wh{s}\in Y}\Supp(\trans(\wh{s},a))\cap
\obsmap^{-1}(o)$; and 
\item  action $a$ is allowed in observation $(Y,B,L)$.
\end{itemize}
\item If an action $a$ is not allowed in observation $(Y,B,L)$,
then we add a transition $(s,(Y,B,L)) \stackrel{a}{\rightarrow}
\eqbad$, i.e., if the conditions are not satisfied
the action leads to the state $\eqbad$ that will be a loosing absorbing state in
the POMDP $\eqmdp$.

 
\item $(s',Y',(Y,B,L),a) \stackrel{(Y',B',L')}{\rightarrow} (s',(Y',B',L'))$ 
iff the action $(Y',B',L')$ is allowed in the observation $(Y',(Y,B,L),a)$.
Again if an action is not allowed, then the transition leads only to $\eqbad$.

\item The state $\eqbad$ is an absorbing state, i.e., $\eqbad
\stackrel{\wh{a}}{\rightarrow} \eqbad$ for all actions $\wh{a} \in \eqact$.
\end{enumerate}
\end{itemize}
Intuitively, $\eqmdp$ allows all possible ways that a projected
strategy of a finite-memory almost-sure winning strategy could 
possibly play in $\game$.
We define the coB\"uchi objective $\coBuchi(\wh{p}^{-1}(2))$ 
with the priority function for the POMDP $\eqmdp$ as 
$\wh{p}((s,m)) = \wh{p}((s,Y,m,a)) = p(s)$. 
The priority for the initial state $\wh{p}(\eqinitial)$ may be set to an 
arbitrary priority from $\{1,2\}$ as the initial 
state will be visited only once. The priority for the state $\eqbad$ is set to~$1$,
i.e., $\wh{p}(\eqbad) = 1$.
We will refer to the above construction as $\almostcobuchi$ construction,
i.e., $\eqmdp=\almostcobuchi(\game)$.
We first argue that $\eqmdp$ is a belief-observation POMDP.

\begin{lemma}
The POMDP $\eqmdp$ is a belief-observation POMDP.
\end{lemma}
\begin{proof}
Note that the observations are defined in a way  that the first component 
cannot be observed. 
Given a sequence $\wh{w}$ of states and actions in $\eqmdp$ with the 
observation sequence $\wh{\rho} = \eqobsmap(\wh{w})$ we  will show that 
the possible first components of the states in the belief $\B(\wh{\rho})$ 
are equal the updated belief components $Y'$ in the observation. 
Intuitively the proof holds as the $Y$ component is the belief and the belief
already represents exactly the set of states in which the POMDP can be with positive 
probability. 
We now present the formal argument.
Let us denote by $\proj_1(\B(\wh{\rho})) \subseteq S$ the projection on the 
first component of the states in the belief. 
One inclusion is trivial since for every reachable state $(s,(Y,B,L))$ we
have  $s\in Y$ (resp.
for states $(s',Y',(Y,B,L),a)$ we have that $s'\in Y'$).  
Therefore we have $\proj_1(\B(\wh{\rho})) \subseteq Y$ (resp. $Y'$).

We prove the second inclusion by induction with respect to the length of the play prefix:
\begin{itemize}
\item \textbf{Base case:} We show the base case for prefixes of length $1$ and
$2$. The first observation is always $\{\eqinitial\}$ which contains only a
single state, so there is nothing to prove. Similarly the second observation in
the game is of the form $(\{s_0\},B,L)$ for some $B \in \{0,1\}^S, L \in \D^S$, and the argument is the same.

\item \textbf{Induction step:} 
Let us a consider a prefix $\wh{w}' = \wh{w} \cdot a \cdot (s',Y',(Y,B,L),a)$ where
$a \in \Act$ and the last transition is $(s,(Y,B,L)) \stackrel{a}{\rightarrow}
(s',Y',(Y,B,L),a)$ in 
the POMDP $\eqmdp$.
By induction hypothesis we have that $\B(\eqobsmap(\wh{w})) 
= \{(s,(Y,B,L)) \mid s\in Y\}$.
The new belief is computed (by definition) as 
\[
\B(\eqobsmap(\wh{w}')) = 
\bigcup_{\wh{s} \in Y}
\Supp(\eqtrans((\wh{s},(Y,B,L)),a)) \cap \eqobsmap^{-1}((Y',(Y,B,L),a)).
\]

 Let $s_{Y'}$ be a state in $Y'$, we want to show that $(s_{Y'},Y',(Y,B,L),a)$ is in
$\B(\eqobsmap(\wh{w}'))$. Due to the definition of the belief update there
exists a state $s_{Y}$ in $Y$ such that $s_{Y} \stackrel{a}{\rightarrow} s_{Y'}$
and $(s_Y,(Y,B,L)) \in \B(\eqobsmap(\wh{w}))$. As $\obsmap(s_{Y'}) =
\obsmap(s')$, it follows that $(s_{Y'},Y',(Y,B,L),a) \in \bigcup_{\wh{s} \in \B(\eqobsmap(\wh{w}))}
\Supp(\eqtrans(\wh{s},a))$ and as  $(s_{Y'},Y',(Y,B,L),a) \in \eqobsmap^{-1}((Y',(Y,B,L),a))$,
the result follows.



The case when the prefix is extended with an memory action $m \in M$ is simpler as the first two components do not change during the transition.

\end{itemize}
The desired result follows.
\end{proof}

The proof of the following two lemmas will use some desired properties of
the projected strategy of a finite-memory strategy and the $\BoolRec$ and
$\SetRec$ functions established in Section~\ref{sec:finmem}. The properties are as follows:
\begin{enumerate}
\item \emph{(Property A for $\BoolRec$ functions).} 
For every run of the Markov chain obtained from the POMDP and a 
projected strategy $\projp(\straa)$ of a finite-memory strategy 
$\straa$, if $\BoolRec$ is set to~1, then for all successors
$\BoolRec$ remains~1 (follows from the fifth point of Lemma~\ref{l-chain_Z}).

\item \emph{(Property B for $\SetRec$ functions).} 
For every run of the Markov chain obtained from the POMDP and a 
projected strategy $\projp(\straa)$ of a finite-memory strategy 
$\straa$, $\SetRec$ functions are non-increasing along the steps
of the run (follows from the third point of Lemma~\ref{l-chain_Z}).

\item \emph{(Property C for $\BoolRec$ and $\SetRec$ functions).} 
For every run of the Markov chain obtained from the POMDP and a 
projected strategy $\projp(\straa)$ of a finite-memory strategy 
$\straa$, if $\BoolRec$ is set to~1 for a state $s$, then 
all reachable states from that point have a priority in 
$\SetRec_{\projp(\straa)}(s)$ (follows from the sixth point of Lemma~\ref{l-chain_Z}).

\end{enumerate}

\begin{lemma}
\label{lem:pomdptoeqmdp}
If there exists a finite-memory almost-sure winning strategy in the POMDP 
$\game$ for the coB\"uchi objective $\coBuchi(p^{-1}(2))$, then there exists a randomized memoryless almost-sure  winning strategy in the belief-observation POMDP 
$\eqmdp$ for the coB\"uchi objective $\coBuchi(\wh{p}^{-1}(2))$.
\end{lemma}

\begin{proof}
Assume there exists a finite-memory almost-sure winning strategy $\straa$, then
by Theorem~\ref{thrm:muller} there exists a finite-memory almost-sure winning
strategy $\straa' = \projp(\straa)$, which uses memory $M' \subseteq M$.
Let $\straa' = (\straa'_u, \straa'_n,M',(\{s_0,\},B_0,L_0))$ be the almost-sure
winning strategy in the POMDP $\game$. 
We fix the strategy $\straa'$ in the POMDP $\game$ and obtain a Markov Chain
$\game_1=\game \restr \straa'$. 
We define a randomized memoryless observation-based strategy $\wh{\straa}:\eqObs
\rightarrow \distr(\eqact)$ in the POMDP $\eqmdp$ 
as follows:

\begin{itemize}
\item 
The deterministic action in the initial observation is:
$\wh{\straa}(\{\eqinitial\}) = (\{s_0\},B_0,L_0)$.
\item In the action-selection observation $(Y,B,L)$ we define
$\wh{\straa}((Y,B,L)) = \straa'_n((Y,B,L))$.

\item In the memory-selection observation $(Y',(Y,B,L),a)$ we define
$\wh{\straa}((Y',(Y,B,L),a))$ to play uniformly actions from the set  
$\Supp((\straa'_u((Y,B,L), o', a)))$, where $o'$ is the unique observation such that
all states in $Y'$ have observation $o'$ in $\game$. 

\item In the observation of the absorbing state $\{\eqbad\}$ no matter what
actions are played the state $\eqbad$ is not left.
\end{itemize}

We fix the memoryless strategy $\wh{\straa}$ in the POMDP $\eqmdp$ and obtain a
Markov chain $\game_2=\eqmdp \restr \wh{\straa}$. For simplicity we will write
$(s,(Y,B,L)) \rightarrow (s',(Y',B',L'))$ whenever $(s,(Y,B,L)) \rightarrow
(s',Y',(Y,B,L),a)
\rightarrow (s',(Y',B',L'))$ for some $a \in \Act$ in $\game_2$. Note that
omitting the intermediate state does not affect the objective as
$\wh{p}((s',Y',(Y,B,L),a)) = 
\wh{p}((s',(Y',B',L')))$. 

The strategy $\wh{\straa}$ will in the first step select the $(\{s_0\},B_0,L_0)$
action and reach state $(s_0,(\{s_0\},B_0,L_0))$ in $\game_2$. We will show that the
two Markov chains reachable from initial state $(s_0,(\{s_0\},B_0,L_0))$ in $\game_1$ and
$(s_0,(\{s_0\},B_0,L_0))$ in $\game_2$ are isomorphic (when considering the simplified edges in $\game_2$).

\begin{itemize}
\item Let $(s,(Y,B,L)) \rightarrow (s',(Y',B',L'))$ be an edge in $\game_1$, then there 
exists (i)~an action $a \in \Supp( {\straa'}_n((Y,B,L)))$, such that 
$s' \in \Supp(\trans(s,a))$, (ii)~$Y'$ is the belief update from $Y$ under 
observation $\obsmap(s')$ and action $a$, and (iii)~$(Y',B',L') \in
\Supp({\straa'}_u((Y,B,L), \obsmap(s'), a))$. 
First we show that there is a transition $(s,(Y,B,L)) \stackrel{a}{\rightarrow}
(s',Y',(Y,B,L),a)$ in the POMDP $\eqmdp$. We verify three properties of the
transition functions $\eqtrans$: The property that 
$s' \in \Supp(\trans(s,a))$ and that $Y'$ is the belief update follows from the
facts (i) and (ii) mentioned above.
Next we show that action $a$ is allowed in observation $(Y,B,L)$, i.e.,  we need to verify that  for all states $\wh{s} \in Y$ such that $B(\wh{s}) = 1$,
and there is a subset of priorities $Z_\infty \subseteq \set{1,2}$ such that 
$L(\wh{s})= \set{Z_\infty}$, and $p(\wh{s})  \in Z_\infty$
we have that all states reachable in one step $\wh{s}' \in \Supp(\trans(\wh{s},a))$ satisfy 
$p(\wh{s}') \in Z_\infty$.
Consider an arbitrary state $\wh{s} \in Y$ such that $B(\wh{s}) = 1$,
$L(\wh{s})=\set{Z_\infty}$ and $p(\wh{s}) \in Z_\infty$. 
Note that  $(\wh{s},(Y,B,L))$ is a reachable pseudo-recurrent state in the Markov chain $G_1$ 
since $B(\wh{s}) = 1$, $L(\wh{s})=\set{Z_\infty}$ and $p(\wh{s}) \in Z_\infty$.
As the strategy $\straa'$ is almost-sure winning it
follows that all recurrent classes reachable from $(\wh{s},(Y,B,L))$ contain
 only states with priority $2$ i.e., $Z_\infty = \{2\}$. 
Moreover by property A and C it follows that only
states with priority in $Z_\infty$  (i.e, with priority $2$) 
are reachable from $(\wh{s},(Y,B,L))$ in $\game_1$. It follows that
action $a$ is allowed in state $(s,(Y,B,L))$.
By the definition of the strategy $\wh{\straa}$ this
action is played with positive probability and therefore $(s,(Y,B,L))
\rightarrow (s',Y',(Y,B,L),a)$ is an edge in $G_2$. Similarly, we show 
that there is a transition $(s',Y',(Y,B,L),a)
\stackrel{(Y',B',L')}{\rightarrow} (s',(Y',B',L'))$ in $\eqmdp$.
To show that
$(Y',B',L')$ is an allowed action in observation $(Y',(Y,B,L),a)$ 
we consider all states  $\wh{s} \in Y$, such that  $B(\wh{s}) = 1$, and all
reachable states $\wh{s}' \in
\Supp(\trans(\wh{s},a))$, and want to show that $B'(\wh{s}') =1$. As the state
$(\wh{s},(Y,B,L))$ is reachable and $\straa'$ is a projected strategy of an almost-sure winning
strategy, it follows by property B and C of
$\BoolRec$ and $\SetRec$ functions of the projected
strategy $\straa'$, i.e., point three and five in Lemma~\ref{l-chain_Z}, that all the memories in the memory
update $ \Supp({\straa'}_u((Y,B,L), \obsmap(s'), a))$ of the projected strategy satisfy that $B'(\wh{s}') = 1$. The second
property of the non-increasing $\SetRec$ function is proved similarly.
It follows that the action satisfies the requirements of item~(4) of the transition 
function $\eqtrans$. 
As before $(Y',B',L')$ is played with positive probability and hence $(s',Y',(Y,B,L),a)
\rightarrow (s',(Y',B',L'))$ is an edge in $\game_2$. It follows that 
$(s,(Y,B,L)) \rightarrow (s',(Y',B',L'))$ is an edge in the simplified graph of 
$\game_2$.

\item In the other direction  let us consider an edge $(s,(Y,B,L)) \rightarrow
(s',(Y',B',L'))$ in the simplified graph of $\game_2$, it follows that there
exists an action $a$, such that there are edges $(s,(Y,B,L)) \rightarrow
(s',Y',(Y,B,L),a) \rightarrow (s',(Y',B',L'))$ in the full graph $\game_2$.
 By the definition of the POMDP $\eqmdp$ we get that $s' \in \Supp(\trans(s,a))$
and $Y'$ is the belief update from $Y$ under observation $\obsmap(s')$ and
action $a$. As action $a$ was played with positive probability it follows by the
definition of the strategy $\wh{\straa}$ that $a \in \Supp
({\straa'}_n((Y,B,L)))$ and similarly $(Y',B',L')$ being played by 
$\wh{\sigma}$ we get that  $(Y',B',L') \in \Supp({\straa'}_u((Y,B,L), \obsmap(s'), a))$.
Hence we get that $(s,(Y,B,L)) \rightarrow (s',(Y',B',L'))$ is an edge in $G_1$.
\end{itemize}
The desired result follows.
\end{proof}

\begin{lemma}
\label{lem:eqmdptopomdp}
If there exists a randomized memoryless almost-sure winning 
strategy in the belief-observation POMDP $\eqmdp$ for the coB\"uchi objective $\coBuchi(\wh{p}^{-1}(2))$, then there exists a 
finite-memory almost-sure winning strategy in the POMDP 
$\game$ for the coB\"uchi objective $\coBuchi(p^{-1}(2))$.
\end{lemma}

\begin{proof}
Given a memoryless strategy $\wh{\straa}$ in $\eqmdp$, we define the
finite-memory strategy $\straa = (\straa_u, \straa_n,M, m_0)$ in $\game$ as follows:
\begin{itemize}
\item $\straa_n((Y,B,L)) = \wh{\straa}((Y,B,L))$;
\item $\straa_u((Y,B,L),o,a)$ update uniformly to elements from the set
$\Supp(\wh{\straa}((Y',(Y,B,L),a)))$, where $Y'$ is the belief update from $Y$ under observation 
$o$ and action $a$.
\item $m_0 = \wh{\straa}(\{\eqinitial\})$. Note that this can be in general a
probability distribution. Since we require the initial memory to be
deterministic, we can model this property by adding an additional initial state and
memory state from which the required randomized memory update is performed.
\end{itemize}

We fix the finite-memory strategy $\straa$ in the POMDP $\game$ to obtain a Markov Chain $\game_1 = \game \restr \straa$ and similarly fixing the memoryless
strategy $\wh{\straa}$ in the POMDP $\eqmdp$ to obtain a Markov Chain $\game_2 = \eqmdp \restr \wh{\straa}$.

As in the previous lemma we will consider a simplified graph $\game_2$ and write
$(s,(Y,B,L)) \rightarrow (s',(Y',B',L'))$ whenever $(s,(Y,B,L)) \rightarrow
(s',Y',(Y,B,L),a)
\rightarrow (s',(Y',B',L'))$ for some $a \in \Act$ in $\game_2$. We show that
the two graphs reachable from states $(s_0, (\{s_0\},B_0,L_0))$ in $\game_1$ and
$(s_0, (\{s_0\},B_0,L_0))$ in $\game_2$ are isomorphic. Note that the absorbing
state $\eqbad$ is not reachable in the Markov chain $\game_2$, otherwise 
there would be reachable recurrent class in $\game_2$  containing
the state $\eqbad$ and no other state (follows from the fact that $\eqbad$ is an
absorbing state). As the priority $\wh{p}(\eqbad)$ is  $1$
it follows that there would be a reachable recurrent class with minimal priority
$1$ and contradicting the assumption that $\wh{\straa}$ is an almost-sure winning
strategy.
\begin{itemize}
\item Let $(s,(Y,B,L)) \rightarrow (s',(Y',B',L'))$ be an edge in $\game_1$,
then there exists an action $a \in \Act$ such that (i)~$a \in
\Supp(\wh{\straa}((Y,B,L))$,
(ii)~$s' \in \Supp(\trans(s,a))$,
 and (iii)~$(Y',B',L') \in \Supp(\wh{\straa}(Y',(Y,B,L),a))$.
Therefore there are edges $(s,(Y,B,L)) \rightarrow (s',Y',(Y,B,L),a)$ and
$(s',Y',(Y,B,L),a) \rightarrow (s',(Y',B',L'))$ in $\game_2$.


\item In the other direction let there be an edge  $(s,(Y,B,L)) \rightarrow
(s',(Y',B',L'))$ in the simplified graph of $G_2$, then there exists an action $a \in \Act$ such that 
$(s,(Y,B,L)) {\rightarrow} (s',Y',(Y,B,L),a)  {\rightarrow} (s',(Y',B',L'))$ are transitions in the full graph of $G_2$.
By the definition of the POMDP $\eqmdp$ and the strategy $\straa$ we get that
(i)~$s' \in \Supp(\trans(s,a))$, (ii)~$a \in \Supp (\straa_n((Y,B,L)))$ and
(iii)~$(Y',B',L') \in \Supp(\straa_u((Y,B,L),\obsmap(s'),a))$.
Therefore $(s,(Y,B,L)) \rightarrow (s',(Y',B',L'))$ is an edge in the graph of
$G_1$.
\end{itemize}
The desired result follows.
\end{proof}

\subsection{Polytime algorithm for belief-observation POMDPs}
In this section we will present a polynomial time algorithm for 
the computation of the almost-sure winning set 
for the belief-observation POMDP $\eqmdp$ for coB\"uchi objectives under 
randomized memoryless strategies.
The algorithm will use solutions of almost-sure winning sets for 
safety and reachability objectives.

\smallskip\noindent{\bf POMDPs with available actions.}
For simplicity in presentation we will consider POMDPs with 
an available action function that maps to every observation
the set of available actions for the observation, i.e., we consider POMDPs as
tuples $(S, \Act, \trans, \Obs, \obsact, \obsmap, s_0)$, where the function 
$\obsact: \Obs \rightarrow 2^{\Act}\setminus \emptyset$ maps every 
observation to a non-empty set of available actions.
Note that this is for simplicity in presentation, as
if an action is not available for an observation, then a new 
state can be added that is loosing and for every unavailable 
action transitions can be added to the newly added loosing state
(thus making all actions available).

\smallskip\noindent{\bf Almost-sure winning observations.} 
For an objective $\varphi$, we denote by 
$\Almost(\varphi)=\set{o \in \Obs \mid \text{there exists a randomized memoryless 
strategy $\sigma$ such that for all } s \in \obsmap^{-1}(o). \ \Prb_s^\sigma(\varphi)=1}$
the set of observations such that there is a randomized memoryless strategy
to ensure winning with probability~1 from all states of the 
observation. 
Our goal is to compute $\Almost(\coBuchi(\wh{p}^{-1}(2)))$.
Also note that since we consider belief-observation POMDPs we can only 
consider beliefs that correspond to all states of an observation.
First we introduce one necessary notation:
\begin{itemize}
\item \emph{(Allow).} Given a set $\obsset \subseteq \Obs$ of observations  and an observation $o \in \obsset$
we define by $\allow(o,\obsset)$ the set of actions that when played in $o$ ensures that the next observation 
is in $\obsset$, i.e., more formally:
$$\allow(o,\obsset) = \{a \in \obsact(o) \mid \bigcup_{s \in \obsmap^{-1}(o)} \obsmap(\Supp(\trans(s,a))) \subseteq \obsset\}.$$
\end{itemize}
We will consider the POMDP $\eqmdp=\almostcobuchi(\game)$ obtained by the construction
for reduction to belief-observation POMDPs.

\begin{definition}
Given the POMDP $\eqmdp$, for a set $F \subseteq \eqstate$ of states, if 
$\{s_0\} \in \Almost(\Safe(F))$, 
we define a POMDP $\eqmdp_{\Safe(F)} = (\eqstate, \Act, \trans, \Almost(\Safe(F)), \eqobsact, \eqobsmap, s_0)$ as follows:
\begin{itemize}
\item The set of states is $\eqstate = \obsmap^{-1}(\Almost(\Safe(F)))$;
\item the available actions are restricted as follows: $\eqobsact(o) = \allow(o,\Almost(\Safe(F)))$; and
\item the observation mapping function $\eqobsmap(s) = \obsmap(s)$.
\end{itemize}
\end{definition}

\begin{lemma}
The POMDP $\eqmdp_{\Safe(F)}$ is a belief-observation POMDP.
\end{lemma}
\begin{proof}
Follows directly from the fact that $\eqmdp$ is a belief-observation POMDP.
\end{proof}

\smallskip\noindent{\bf Almost-sure winning for coB\"uchi objectives.} 
In this part we will show how to decide whether an observation $o \in
\eqObs$ is an
almost-sure winning observation for the coB\"uchi objective
$\coBuchi(\wh{p}^{-1}(2))$ in the 
belief-observation POMDP $\eqmdp$ 
(i.e., whether $o \in \Almost(\coBuchi(\wh{p}^{-1}(2)))$). 
We will show that the computation can be achieved by computing almost-sure winning regions
for safety and reachability objectives.
The steps of the computation are as follows:
\begin{enumerate}
\item \emph{(Step~1).} 
Let $F=\eqstate \setminus \eqbad$ and we first compute $\eqmdp_{\Safe(F)}$. 
This step requires the computation of the almost-sure winning for safety objectives.

\item \emph{(Step~2).} 
Let $\eqstate_{\wpr} \subseteq \eqstate$ denote the subset of states
that intuitively correspond to \emph{winning pseudo-recurrent (wpr)}  states, 
i.e., formally it is defined as follows:
$$\eqstate_{\wpr} = \{ (s,(Y,B,L)) \mid B(s) = 1,  L(s) = \set{\set{2}} \mbox { and }
\wh{p}(s)=2 \}.$$
In the restricted POMDP $\eqmdp_{\Safe(F)}$ we compute the set of observations 
$W_2 =\Almost(\Reach(\eqstate_{\wpr}))$. 
We will show that $W_2=\Almost(\coBuchi(\wh{p}^{-1}(2)))$.
This step requires the computation of the almost-sure winning for reachability 
objectives.
\end{enumerate}
In the following two lemmas we establish the two required inclusions to show 
$W_2 = \Almost(\coBuchi(\wh{p}^{-1}(2)))$.

\begin{lemma}
$W_2 \subseteq \Almost(\coBuchi(\wh{p}^{-1}(2)))$.
\end{lemma}

\begin{proof}
Let $o \in W_2$ be an observation in $W_2$, and we show how to construct a
randomized memoryless almost-sure winning strategy ensuring that $o \in
\Almost(\coBuchi(\wh{p}^{-1}(2)))$. 
Let $\straa$ be the strategy produced by the computation of
$\Almost(\Reach(\eqstate_{\wpr}))$. We will show that the same strategy ensures
also $\Almost(\coBuchi(\wh{p}^{-1}(2)))$. As in every observation $o$ the
strategy $\straa$ plays only a subset of actions that are in
$\allow(o,\Almost(\Safe(F))$ (to ensure safety in $F$),
where $F= \eqstate \setminus \eqbad$, the absorbing state $\eqbad$ is not reachable. 
Also with probability~$1$ the set $\eqstate_{\wpr}$ is reached.
We show that for all states $(s,(Y,B,L)) \in \eqstate_{\wpr}$ that 
all the states reachable from $(s,(Y,B,L))$ have priority $2$ according to $\wh{p}$. 
Therefore ensuring that all recurrent classes reachable from $\eqstate_{\wpr}$ 
have minimal priority $2$.
Due to the construction of the POMDP $\eqmdp$, the only actions allowed in a
state $(s,(Y,B,L))$ 
satisfy that for all states $\wh{s} \in Y$ if $B(\wh{s}) = 1$, $L(\wh{s}) =
\set{Z_\infty}$ and $\wh{p}(s) \in Z_\infty$ for some $Z_\infty \subseteq
\set{1,2}$, then for all states $\wh{s}' \in
\Supp(\trans(s,a))$ we have that $p(\wh{s}') \in Z_\infty$.
 As all states in $(s,(Y,B,L)) \in
\eqstate_{\wpr}$ have $L(s) = \set{\set{2}}$, it follows that any state 
reachable in the next step has priority $2$.
Let $(s',Y',(Y,B,L),a)$ be an arbitrary state reachable from $(s,(Y,B,L))$ in one step. 
By the previous argument we have that the priority $\wh{p}((s',Y',(Y,B,L),a)) = 2$.  
Similarly the only allowed memory-update actions $(Y',B',L')$ from state $(s',Y',(Y,B,L),a)$ 
satisfy that whenever $\wh{s} \in Y$ and $B(\wh{s}) = 1$, then for all $\wh{s}' \in
\Supp(\trans(\wh{s},a))$, we have that $B'(\wh{s}') =1$ and similarly we have
that $L'(s')$ is a non-empty subset of $L(s)$, i.e., $L'(s') = \set{\set{2}}$. Therefore the next
reachable state $(s',(Y',B',L'))$ is again in $\eqstate_\wpr$.
In other words, from states $(s,(Y,B,L))$ in $\eqstate_{\wpr}$
in all future steps only states with priority~2 are visited, i.e., $\Safe(\wh{p}^{-1}(2))$
is ensured which ensures the coB\"uchi objective.
As the states in $\eqstate_{\wpr}$ are reached with probability $1$ and from them
all recurrent classes reachable have only states that have priority $2$,
the desired result follows.
\end{proof}

\begin{lemma}
$\Almost(\coBuchi(\wh{p}^{-1}(2))) \subseteq W_2$.
\end{lemma}
\begin{proof}
Assume towards contradiction that there is an observation $o$ in $\eqObs
\setminus W_2$ such that $o \in \Almost(\coBuchi(\wh{p}^{-1}(2)))$. The
observation $o$ belongs to $\Almost(\Safe(F))$ as there is no winning strategy
from observations outside $\Almost(\Safe(F))$, where $F= \eqstate \setminus \eqbad$.
Consider a  randomized memoryless strategy $\wh{\straa}$ satisfying the coB\"uchi objective
from  all observations in $\Almost(\coBuchi(\wh{p}^{-1}(2)))$.
By Lemma~\ref{lem:eqmdptopomdp} there exists a finite-memory almost-sure winning strategy $\straa$ in the POMDP
$\game$. Let us consider the almost-sure winning projected strategy $\straa' = \projp(\straa)$ in
the POMDP $\game$. Recall that by Lemma~\ref{l-reach-pr-prob-1} the set of
pseudo-recurrent states is reached with probability $1$ in the Markov chain
$\game \restr \straa'$. 
By Lemma~\ref{l-rec_pr_pr}, for a pseudo-recurrent state $(s,(Y,B,L))$
there exists reachable recurrent class with the priority set as $L(s)$, and 
since $\straa'$ is an almost-sure winning strategy every recurrent class
must have only priority~2, and hence for every reachable pseudo-recurrent state 
$(s,(Y,B,L))$ we must have $L(s) =\set{\set{2}}$. And by the definition of
pseudo-recurrent states we also have that $B(s) = 1$ and $p(s) \in \set{2}$. As
$\wh{p}(s,(Y,B,L)) = p(s)$ we have that $(s,(Y,B,L)) \in \eqstate_{\wpr}$.
This implies that every pseudo-recurrent state reachable is a state in 
$\eqstate_{\wpr}$.
We want to show that in the construction described in
Lemma~\ref{lem:pomdptoeqmdp}, the memoryless almost-sure winning strategy
$\wh{\straa}'$ constructed from the projected strategy $\straa'$ will ensure 
reaching the set $\wh{S}_{\wpr}$ in the Markov chain 
$\eqmdp \restr \wh{\straa}'$ with probability $1$. In the proof of 
Lemma~\ref{lem:pomdptoeqmdp} we have already established that
reachability is preserved, i.e., if $(s',(Y',B',L'))$ is reachable from
$(s,(Y,B,L))$ in $\game \restr \straa'$ then $(s',(Y',B',L'))$ is reachable from
$(s,(Y,B,L))$ in $\eqmdp \restr \wh{\straa}'$. As by Lemma~\ref{l-reach-pr-prob-1} from every 
state a pseudo-recurrent state is reached with positive probability, and 
(as argued above every reachable pseudo-recurrent state is in $\eqstate_{\wpr}$) 
we have that from every state in $\eqmdp \restr \wh{\straa}'$ a state in $\wh{S}_{\wpr}$ is reachable. 
As this is true for every state we have that the set of states $\wh{S}_{\wpr}$ is reached
with probability $1$ in $\eqmdp \restr \wh{\straa}'$ (Property~1~(a)). Therefore we have that the
observation $o$ belongs to $\Almost(\Reach(\eqstate_{\wpr}))$.
But this contradicts that $o$ does not belong to $W_2$ and the desired result follows.
\end{proof}

To complete the computation for almost-sure winning for coB\"uchi objectives
we now present polynomial time solutions for almost-sure safety and 
almost-sure B\"uchi objectives (that implies the solution for almost-sure
reachability) in belief-observation POMDPs for 
randomized memoryless strategies.
The algorithm is presented in~\cite{CC13} and we present them below just for 
sake of completeness.
We start with a few notations below:

\begin{itemize}

\item \emph{(Pre).} The predecessor function given a set of observations $\obsset$ selects the observations $o \in \obsset$ such that $\allow(o,\obsset)$ is non-empty , i.e.,
$$\pre(\obsset) = \{ o \in \obsset \mid \allow(o,\obsset) \not = \emptyset \}.$$

\item \emph{(Apre).} Given a set $Y \subseteq \Obs$ of observations  and a set
$X \subseteq S$ of states such that $X \subseteq \obsmap^{-1}(Y)$, the set $\apre(Y,X)$ denotes the states from $\obsmap^{-1}(Y)$ such that there exists an action that ensures that the next observation is in $Y$ and the set $X$ is reached with positive probability, i.e.,:
$$ \apre(Y,X) = \{ s \in \obsmap^{-1}(Y) \mid \exists a \in \allow(\obsmap(s),Y) \text{ such that } \Supp(\trans(s,a)) \cap X \not =  \emptyset \}.$$
\item \emph{(ObsCover).} For a set $U \subseteq S$ of states we define the $\obscover(U) \subseteq \Obs$ to be the set of observations $o$ such that all states
with observation $o$ is in $U$, 
i.e., $\obscover(U) = \{ o \in \Obs \mid \obsmap^{-1}(o) \subseteq U\}$.
\end{itemize}
Using the above notations we present the solution of almost-sure 
winning for safety and B\"uchi objectives.

\smallskip\noindent{\bf Almost-sure winning for safety objectives.} 
Given a safety objective $\Safe(F)$, for a set $F \subseteq S$ of states, 
let $\obsset_F=\obscover(F)$ denote the set of observations $o$ such that 
$\gamma^{-1}(o) \subseteq F$, i.e., for all states $s \in \obsmap^{-1}(o)$ 
belong to $F$.
We denote by $\nu X$ the greatest fixpoint and by $\mu X$ the least fixpoint. 
Let 
\[
Y^* = \nu Y.(\obsset_F \cap \pre(Y)) =\nu Y. (\obscover(F) \cap \pre(Y))
\]
be the greatest fixpoint of the function $f(Y)= \obsset_F \cap \pre(Y)$. 
Then the set $Y^*$ is obtained by the following computation:
\begin{enumerate}
\item $Y_0 \leftarrow \obsset_F$; and
\item repeat $Y_{i+1} \leftarrow \pre(Y_{i})$ until a fixpoint is reached.
\end{enumerate}
We show that $Y^*=\Almost(\Safe(F))$.

\begin{lemma}
\label{lem:nonemptyallow}
For every observation $o \in Y^*$ we have $\allow(o,Y^*)\neq \emptyset$ 
(i.e., $\allow(o,Y^*)$ is non-empty).
\end{lemma}
\begin{proof}
Assume towards contradiction that there exists an observation $o \in Y^*$ 
such that $\allow(o,Y^*)$ is empty. 
Then $o \not\in \pre(Y^*)$ and hence the observation 
must be removed in the next iteration of the algorithm. 
This implies $\pre(Y^*) \not = Y^*$, we reach a contradiction that 
$Y^*$ is a fixpoint.
\end{proof}

\begin{lemma}
The set $Y^*$ is the set of almost-sure winning observations for the
safety objective $\Safe(F)$, i.e., $Y^* = \Almost(\Safe(F))$, 
and can be computed in linear time.
\end{lemma}

\begin{proof}
We prove the two desired inclusions:
(1)~$Y^* \subseteq \Almost(\Safe(F))$; and 
(2)~$\Almost(\Safe(F)) \subseteq Y^*$.
\begin{enumerate}
\item \emph{(First inclusion).} By the definition of $Y_0$ we have that 
$\obsmap^{-1}(Y_0) \subseteq F$. 
As $Y_{i+1} \subseteq Y_{i}$ we have that $\obsmap^{-1}(Y^*) \subseteq F$. 
By Lemma~\ref{lem:nonemptyallow}, for all observations $o\in Y^*$ we have 
$\allow(o,Y^*)$ is non-empty. 
A pure memoryless that plays some action from $\allow(o,Y^*)$ in $o$, for 
$o\in Y^*$, ensures that the next observation is in $Y^*$.
Thus the strategy ensures that only states from 
$\obsmap^{-1}(Y^*) \subseteq F$ are visited, 
and therefore is an almost-sure winning strategy for the safety objective.

\item \emph{(Second inclusion).} We prove that there is no almost-sure winning
strategy from $\Obs \setminus Y^*$ by induction:
\begin{itemize}
\item \textbf{(Base case).} There is no almost-sure winning strategy from observations $\Obs \setminus Y_0$. Note that $Y_0 = \obsset_F$.
In every observation $o \in \Obs \setminus Y_0$ there exists a state $s \in \obsmap^{-1}(o)$ such that $s \not \in F$. As $\game$ is a belief-observation POMDP there is a positive probability of being in state $s$, and therefore not being in $F$. 
\item \textbf{(Inductive step).} We show that there is no almost-sure winning strategy from observations in $\Obs \setminus Y_{i+1}$. Let $Y_{i+1} \not = Y_{i}$ and $o \in Y_{i} \setminus Y_{i+1}$ (or equivalently $(\Obs \setminus Y_{i+1}) \setminus (\Obs \setminus Y_{i})$). As the observation $o$ is removed from $Y_i$ it follows that $\allow(o, Y_i) = \emptyset$. It follows that no matter what action is played, there is a positive probability of being in a state $s \in \obsmap^{-1}(o)$ such that playing the action would leave the set $\obsmap^{-1}(Y_i)$ with positive probability, and thus reaching the observations $\Obs \setminus Y_i$ from which there is no almost-sure winning strategy by induction hypothesis.
\end{itemize}
\end{enumerate}
This shows that $Y^*=\Almost(\Safe(F))$, and the linear time computation follows
from the straight forward computation of greatest fixpoints.
The desired result follows.
\end{proof}
We now present one simple lemma that was implicitly used in the restriction of
the POMDP $\eqmdp$ to almost-sure safety that a randomized memoryless strategies
must only play action in the $\allow$ set.

\begin{lemma}
\label{lem:safe}
Let $\straa$ be a randomized memoryless almost-sure winning strategy in a belief-observation POMDP $\game$ for the safety objective $\Safe(F)$. Then $\Supp(\straa(o)) \subseteq \allow(o,\Almost(\Safe(F)))$.
\end{lemma}
\begin{proof}
Assume that the strategy $\straa$ plays  an action $a \in \Act \setminus 
\allow(o,\Almost(\Safe(F)))$ after an observation $o$. Then there is a 
positive probability of being in a state $s \in \obsmap^{-1}(o)$ such that 
playing the action $a$ in that state would leave the observations 
$\Almost(\Safe(F))$ with positive probability. 
As there is no randomized almost-sure winning strategy in 
$S \setminus \Almost(\Safe(F))$ (by definition), 
this contradicts the fact that $\straa$ is almost-sure winning.
\end{proof}

\smallskip\noindent{\bf Almost-sure winning for B\"uchi objectives.} 
Consider a set $T\subseteq S$ of target states, and the B\"uchi objective
$\Buchi(T)$. 
We will show that: 
$$\Almost(\Buchi(T)) = 
\nu Z. \obscover(\mu X. ((T \cap \obsmap^{-1}(Z) \cap \obsmap^{-1}(\pre(Z))) \cup \apre(Z,X))).$$

Let $Z^* = \nu Z. \obscover(\mu X. ((T \cap \obsmap^{-1}(Z) \cap
\obsmap^{-1}(\pre(Z))) 
\cup \apre(Z,X)))$. 
In the following two lemmas we show the two desired inclusions, i.e., 
$\Almost(\Buchi(T)) \subseteq Z^*$ and then we show that 
$Z^* \subseteq \Almost(\Buchi(T))$.

\begin{lemma}
\label{lem:bsuby}
$\Almost(\Buchi(T)) \subseteq Z^*$.
\end{lemma}
\begin{proof}
Let $W^* = \Almost(\Buchi(T))$.
We first show that $W^*$ is a fixpoint of the function 
\[
f(Z)= \obscover (\mu X. ((T \cap \obsmap^{-1}(Z) \cap \obsmap^{-1}(\pre(Z))) \cup
\apre(Z,X))),
\] 
i.e., 
we will show that $W^*= \obscover(\mu X. ((T \cap \obsmap^{-1}(W^*) \cap \obsmap^{-1}(\pre(W^*))) \cup \apre(W^*,X)))$ . 
As $Z^*$ is the greatest fixpoint it will follow that $W^* \subseteq Z^*$.

Let 
\[
X^* = (\mu X.((T \cap \obsmap^{-1}(W^*) \cap
\obsmap^{-1}(\pre(W^*))) \cup
\apre(W^*,X))),
\]
and $\wh{X}^*=\obscover(X^*)$.
Note that by definition we have $X^* \subseteq \obsmap^{-1}(W^*)$ as the inner 
fixpoint computation only computes states that belong to $\obsmap^{-1}(W^*)$.
Assume towards contradiction that $W^*$ is not a fixpoint, i.e., 
$\wh{X}^*$ is a strict subset of $W^*$. 
For all states $s\in \obsmap^{-1}(W^*) \setminus X^*$, 
for all actions $a \in \allow(\obsmap(s),W^*)$ we have 
$\Supp(\trans(s,a)) \subseteq (\obsmap^{-1}(W^*)\setminus X^*)$. 
Consider any randomized memoryless almost-sure winning strategy $\straa^*$ 
from $W^*$ and we consider two cases:
\begin{enumerate}
\item Suppose there is a state $s \in \obsmap^{-1}(W^*)\setminus X^*$ 
such that an action that does not belong to $\allow(\obsmap(s),W^*)$ 
is played with positive probability by $\straa^*$.
Then with positive probability the observations from $W^*$ are left 
(because from some state with same observation as $s$ an observation in the
complement of $W^*$ is reached with positive probability).
Since from the complement of $W^*$ there is no randomized memoryless almost-sure winning 
strategy (by definition), it contradicts that $\straa^*$ is an almost-sure winning 
strategy from $W^*$.
\item Otherwise for all states $s \in \obsmap^{-1}(W^*)\setminus X^*$ 
the strategy $\straa^*$ plays only actions in $\allow(\obsmap(s),W^*)$, 
and then the probability to reach $X^*$ is zero, i.e., $\Safe(\obsmap^{-1}(W^*)\setminus X^*)$ 
is ensured. 
Since all target states in $\obsmap^{-1}(W^*)$ belong to $X^*$ 
(they get included in iteration~0 of the fixpoint computation) 
it follows that $(\obsmap^{-1}(W^*)\setminus X^*) \cap T =\emptyset$, and 
hence $\Safe(\obsmap^{-1}(W^*) \setminus X^*)
\cap \Buchi(T)=\emptyset$,
and we again reach a contradiction that $\straa^*$ is an almost-sure
winning strategy. 
\end{enumerate}
It follows that $W^*$ is a fixpoint, and thus we get that $W^* \subseteq Z^*$.
\end{proof}



\begin{lemma}
\label{lem:ysubb}
$Z^* \subseteq \Almost(\Buchi(T))$.
\end{lemma}
\begin{proof}
We define a randomized memoryless strategy $\straa^*$ for the objective $\Almost(\Buchi(T))$ 
as follows: for an observation $o \in Z^*$, play all actions from the set 
$\allow(o, Z^*)$ uniformly at random.
Since the strategy $\straa^*$ plays only actions
in $\allow(o, Z^*)$, for $o \in Z^*$, it ensures that the set of states 
$\obsmap^{-1}(Z^*)$ is not left, (i.e., $\Safe(\obsmap^{-1}(Z^*))$ is ensured).
We now analyze the computation of the inner fixpoint, i.e., analyze the computation 
of $\mu X.((T \cap \obsmap^{-1}(Z^*) \cap \obsmap^{-1}(\pre(Z^*))) \cup \apre(Z^*,X)))$
as follows:
\begin{itemize}
\item $X_0 = 
(T \cap \obsmap^{-1}(Z^*) \cap \obsmap^{-1}(\pre(Z^*))) \cup \apre(Z^*,\emptyset)))
= T \cap \obsmap^{-1}(Z^*) \cap \obsmap^{-1}(\pre(Z^*)) 
\subseteq T$ (since  $\apre(Z^*,\emptyset)$ is emptyset);
\item $X_{i+1} = (T \cap \obsmap^{-1}(Z^*) \cap \obsmap^{-1}(\pre(Z^*))) \cup
\apre(Z^*,X_i)))$
\end{itemize}
Note that we have $X_0 \subseteq T$.
For every state $s_j \in X_j$ the set of played actions $\allow(\obsmap(s_j),
Z^*)$ contains an action $a$ such that $\Supp(\trans(s_j,a)) \cap X_{j-1}$ is non-empty.
Let $C$ be an arbitrary reachable recurrent class in the Markov Chain $\game \restr
\straa$ reachable from a state in $\obsmap^{-1}(Z^*)$.
Since  $\Safe(\obsmap^{-1}(Z^*))$ is ensured, it follows that $C \subseteq \obsmap^{-1}(Z^*)$.
Consider a state in $C$ that belongs to $X_j \setminus X_{j-1}$ for $j \geq 1$.
Since the strategy ensures that for some action $a$ played with positive 
probability we must have $\Supp(\trans(s_j,a)) \cap X_{j-1} \neq \emptyset$,
it follows that $C \cap X_{j-1} \neq \emptyset$.
Hence by induction $C \cap X_0 \neq \emptyset$.
It follows $C \cap T \neq \emptyset$.
Hence all reachable recurrent classes intersect with the target 
set and thus the strategy $\straa^*$ ensures that $T$ is visited
infinitely often with probability~1.
Thus we have $Z^* \subseteq \Almost(\Buchi(T))$. 
\end{proof}

\begin{lemma}
The set $\Almost(\Buchi(T))$ and $\Almost(\Reach(T))$ can be computed 
in quadratic time for belief-observation POMDPs, for target set $T \subseteq S$.
\end{lemma}
\begin{proof}
For $\Almost(\Buchi(T))$ it follows directly from Lemma~\ref{lem:bsuby} and
Lemma~\ref{lem:ysubb}. The result for $\Almost(\Reach(T))$ follows from the fact
that $\Reach(T)$ is a special case of $\Buchi(T)$ (by converting states in 
the target set $T$ to absorbing states).
\end{proof}

\smallskip\noindent{\bf The EXPTIME-completeness.}
In Section~\ref{sec:reduction} we have established a polynomial time reduction of
POMDPs with parity objectives to POMDPs with coB\"uchi objectives for almost-sure
winning under finite-memory strategies.
In this section we first showed that given a POMDP $\game$ with a coB\"uchi
objective we can construct an exponential size belief-observation POMDP
$\eqmdp$ and the computation of the almost-sure winning set for 
coB\"uchi objectives reduced to the computation of the almost-sure 
winning set for safety and reachability objectives, for which we established
linear and quadratic time algorithms respectively.
This gives us an $2^{O(|S|\cdot d)}$ time algorithm to decide 
(and construct if one exists) the existence of finite-memory almost-sure
winning strategies in POMDPs with parity objectives with $d$ priorities.
The EXPTIME-hardness follows from the results of~\cite{CDH10a} that shows
deciding the existence of finite-memory almost-sure winning strategies  
in POMDPs with reachability objectives is EXPTIME-hard.
The results for positive winning goes via reduction to B\"uchi objectives
and is similar.
We have the following result.

\begin{theorem}
The following assertions hold:
\begin{enumerate}
\item Given a POMDP $G$ with $|S|$ states and a parity objective with $d$ 
priorities, the decision problem of the existence (and the construction if 
one exists) of a finite-memory almost-sure (resp. positive) winning strategy 
can be solved in  $2^{O(|S|\cdot d)}$ time.

\item The decision problem of given a POMDP and a parity objective whether
there exists a finite-memory almost-sure (resp. positive) winning strategy
is EXPTIME-complete.
\end{enumerate}

\end{theorem}

\begin{remark}
Note that our EXPTIME-algorithm for parity objectives, and the LAR reduction
of Muller objectives to parity objectives~\cite{GH82} give 
an $2^{O(d! \cdot d^2 \cdot |S|)}$ time  algorithm for Muller objectives with 
$d$ colors for POMDPs with $|S|$ states,
i.e., the algorithm is exponential in $|S|$ and double exponential in $d$.
Note that the Muller objective specified by the set $\calf$ maybe in general itself 
double exponential in $d$.
\end{remark}

\bibliographystyle{plain}
\bibliography{diss}

\end{document}